%% file: quantum-covert-estimation__jsait_.tex
\documentclass[12pt, onecolumn,draftcls]{IEEEtran}
\usepackage{algorithm, paralist}
\usepackage{algorithmic}
\usepackage[margin=1in]{geometry}
\usepackage{amsmath,stmaryrd,acronym,amssymb,amsthm,dsfont,graphicx}
\usepackage{tikz,etex}
\usetikzlibrary{arrows,arrows.meta,automata,shapes,calc,intersections}
\graphicspath{{graphics/}}
\newcommand{\ket}[1]{|#1\rangle}
\newcommand{\bra}[1]{\langle #1 |}
\newcommand{\braket}[2]{\langle #1 | #2 \rangle}
\newcommand{\set}[1]{{\left\{#1\right\}}}
\newcommand{\supp}[1]{{\mathrm{supp}\pr{#1}}}

\acrodef{AEP}{Asymptotic Equipartition Property}
\acrodef{AoA}{Angle of Arrival}
\acrodef{AWGN}{Additive White Gaussian Noise}
\acrodef{BER}{Bit-Error-Rate}
\acrodef{BEC}{Binary Erasure Channel}
\acrodef{BPSK}{Binary Phase-Shift Keying}
\acrodef{BSC}{Binary Symmetric Channel}
\acrodef{CDF}[CDF]{Cumulative Distribution Function}
\acrodef{CLT}[CLT]{Central Limit Theorem}
\acrodef{CSI}[CSI]{Channel State Information}
\acrodef{DMC}[DMC]{Discrete Memoryless Channel}
\acrodef{DMS}[DMS]{Discrete Memoryless Source}
\acrodef{iid}[i.i.d.]{independent and identically distributed}
\acrodef{lhs}[l.h.s.]{left-hand-side}
\acrodef{rhs}[r.h.s.]{right-hand-side}
\acrodef{LPD}[LPD]{Low Probability of Detection}
\acrodef{LDPC}[LDPC]{Low-Density Parity-Check}
\acrodef{MAC}[MAC]{multiple-access channel}
\acrodef{MIMO}[MIMO]{Multiple-Input Multiple-Output}
\acrodef{MISO}{Multiple-Input Single-Output}
\acrodef{MLC}[MLC]{MultiLevel Coding}
\acrodef{PDF}[PDF]{Probability Distribution Function}
\acrodef{PMF}[PMF]{Probability Mass Function}
\acrodef{PPM}[PPM]{Pulse Position Modulation}
\acrodef{PSD}{Power Spectral Density}
\acrodef{QKD}{Quantum Key Distribution}
\acrodef{QPSK}{Quadrature Phase-Shift Keying}
\acrodef{SIMO}{Single-Input Multiple-Output}
\acrodef{SNR}{Signal-to-Noise Ratio}
\acrodef{wrt}[w.r.t.]{with respect to}
\acrodef{WSS}{Wide Sense Stationary}

\DeclareMathAlphabet{\eurm}{U}{eur}{m}{n}
\DeclareMathAlphabet{\mathbsf}{OT1}{cmss}{bx}{n}% bold sans serif
\DeclareMathAlphabet{\mathssf}{OT1}{cmss}{m}{sl}% slanted sans serif
\DeclareMathAlphabet{\mathcsf}{OT1}{cmss}{sbc}{n}% condensed sans serif

\DeclareSymbolFont{bsfletters}{OT1}{cmss}{bx}{n}  
\DeclareSymbolFont{ssfletters}{OT1}{cmss}{m}{n}
\DeclareMathSymbol{\bsfGamma}{0}{bsfletters}{'000}
\DeclareMathSymbol{\ssfGamma}{0}{ssfletters}{'000}
\DeclareMathSymbol{\bsfDelta}{0}{bsfletters}{'001}
\DeclareMathSymbol{\ssfDelta}{0}{ssfletters}{'001}
\DeclareMathSymbol{\bsfTheta}{0}{bsfletters}{'002}
\DeclareMathSymbol{\ssfTheta}{0}{ssfletters}{'002}
\DeclareMathSymbol{\bsfLambda}{0}{bsfletters}{'003}
\DeclareMathSymbol{\ssfLambda}{0}{ssfletters}{'003}
\DeclareMathSymbol{\bsfXi}{0}{bsfletters}{'004}
\DeclareMathSymbol{\ssfXi}{0}{ssfletters}{'004}
\DeclareMathSymbol{\bsfPi}{0}{bsfletters}{'005}
\DeclareMathSymbol{\ssfPi}{0}{ssfletters}{'005}
\DeclareMathSymbol{\bsfSigma}{0}{bsfletters}{'006}
\DeclareMathSymbol{\ssfSigma}{0}{ssfletters}{'006}
\DeclareMathSymbol{\bsfUpsilon}{0}{bsfletters}{'007}
\DeclareMathSymbol{\ssfUpsilon}{0}{ssfletters}{'007}
\DeclareMathSymbol{\bsfPhi}{0}{bsfletters}{'010}
\DeclareMathSymbol{\ssfPhi}{0}{ssfletters}{'010}
\DeclareMathSymbol{\bsfPsi}{0}{bsfletters}{'011}
\DeclareMathSymbol{\ssfPsi}{0}{ssfletters}{'011}
\DeclareMathSymbol{\bsfOmega}{0}{bsfletters}{'012}
\DeclareMathSymbol{\ssfOmega}{0}{ssfletters}{'012}

\newcommand{\calA}{{\mathcal{A}}}
\newcommand{\calB}{{\mathcal{B}}}

\newcommand{\calD}{{\mathcal{D}}}
\newcommand{\calE}{{\mathcal{E}}}

\newcommand{\calH}{{\mathcal{H}}}

\newcommand{\calL}{{\mathcal{L}}}
\newcommand{\calM}{{\mathcal{M}}}
\newcommand{\calN}{{\mathcal{N}}}
\newcommand{\calO}{{\mathcal{O}}}
\newcommand{\calP}{{\mathcal{P}}}
\newcommand{\calQ}{{\mathcal{Q}}}

\newcommand{\calT}{{\mathcal{T}}}
\newcommand{\calS}{{\mathcal{S}}}
\newcommand{\calU}{{\mathcal{U}}}
\newcommand{\calV}{{\mathcal{V}}}
\newcommand{\calX}{{\mathcal{X}}}
\newcommand{\calY}{{\mathcal{Y}}}

\newcommand{\bfB}{{\mathbf{B}}}
\newcommand{\bfU}{{\mathbf{U}}}
\newcommand{\bfV}{{\mathbf{V}}}
\newcommand{\bfv}{{\mathbf{v}}}
\newcommand{\bfu}{{\mathbf{u}}}
\newcommand{\bfy}{{\mathbf{y}}}
\newcommand{\bfW}{{\mathbf{W}}}
\newcommand{\argmin}{{\textnormal{argmin}}}

\renewcommand{\P}[2][]{{\mathbb{P}_{#1}}{\left(#2\right)}}
\newcommand{\D}[2]{{{\mathbb{D}}\!\left({#1\Vert#2}\right)}}

\newcommand{\Hb}[1]{{\mathbb{H}_b}\left(#1\right)}

\newcommand{\card}[1]{\ensuremath{\left|{#1}\right|}}           % Cardinality
\newcommand{\abs}[1]{\ensuremath{\left|#1\right|}}              % Absolute value        
\newcommand{\norm}[2][]{\ensuremath{{\left\Vert{#2}\right\Vert}_{#1}}}   % Norm
\newcommand{\eqdef}{\ensuremath{\triangleq}}                    % Defined by equality
\newcommand{\intseq}[2]{\ensuremath{\llbracket{#1},{#2}\rrbracket}}  % Sequence of integers
\newcommand{\indic}[1]{\ensuremath{\mathds{1}\!\left\{#1\right\}}}

\renewcommand{\leq}{\leqslant}
\renewcommand{\geq}{\geqslant}

\newcommand{\tr}[1]{\ensuremath{\text{\textnormal{tr}}\left(#1\right)}}  % Trace
\renewcommand{\ker}[1]{\ensuremath{\text{\textnormal{Ker}}\left(#1\right)}}  % Kernel

\newcommand{\proddist}{%
  \mathchoice{\raisebox{1pt}{$\displaystyle\otimes$}}
             {\raisebox{1pt}{$\otimes$}}
             {\raisebox{0.5pt}{\scalebox{0.7}{$\scriptstyle\otimes$}}}
             {\raisebox{0.4pt}{\scalebox{0.6}{$\scriptscriptstyle\otimes$}}}}
\newcommand{\pn}{{\proddist n}}

\usepackage[textwidth=0.68in,textsize=footnotesize,disable]{todonotes}
\presetkeys{todonotes}{color=redArcom, linecolor=redArcom,bordercolor=redArcom}{}

\acrodef{ROC}[ROC]{Receiver Operation Characteristic}
\acrodef{PPM}[PPM]{Pulse-Position Modulation}

\presetkeys{todonotes}{color=redArcom, linecolor=redArcom,bordercolor=redArcom}{}

\acrodef{AVC}{arbitrary varying channel}

\usepackage{xcolor}

\usepackage{stmaryrd,acronym,amsthm,dsfont,graphicx,enumitem}

\newtheorem{theorem}{Theorem}
\newtheorem{lemma}{Lemma}
\newtheorem{definition}{Definition}
\newtheorem{proposition}{Proposition}
\newtheorem{remark}{Remark}

\input{mehrdad_commands.tex}

\newcommand{\Dcc}[3]{\mathbb{C}\pr{#1\|#2|#3}}

\newcommand{\qy}[1]{\rho_{B|\theta}^{#1}}
\newcommand{\qz}[1]{\rho_{W|\theta}^{#1}}

\newcommand{\qpy}[1]{\rho_{\bfB|\theta}^{#1}}
\newcommand{\qpz}[1]{\rho_{\bfW|\theta}^{#1}}

\newcommand{\qyp}[1]{\rho_{B|\theta'}^{#1}}

\newcommand{\qpyp}[1]{\rho_{\bfB|\theta'}^{#1}}

\title{On Covert Quantum Sensing and the Benefits of Entanglement}

\author{Mehrdad Tahmasbi and Matthieu R. Bloch}
\begin{document}
\maketitle

\begin{abstract}
  Motivated by applications to covert quantum radar, we analyze a  covert quantum sensing problem, in which a legitimate user aims at estimating an unknown parameter taking finitely many values by probing a quantum channel while remaining undetectable from an adversary receiving the probing signals through another quantum channel. When channels are classical-quantum, we characterize the optimal error exponent under a covertness constraint for sensing strategies in which probing signals do not depend on past observations. When the legitimate user's channel is a unitary depending on the unknown parameter, we provide achievability and converse results that show how one can significantly improve covertness using an entangled input state. %When the actions are allowed to depend on the past observations, we provide an example showing the benefits of adaptivity.
\end{abstract}

\section{Introduction}
\label{sec:introduction}
While much of the information-theoretic security literature focuses on ensuring secrecy and privacy, in the sense of preventing or minimizing  the information content leaked by signals, there have been recent efforts geared at understanding the information-theoretic limits of covertness, defined as the ability to avoid detection by hiding the mere presence of signals themselves. In particular, such information-theoretic limits have been successfully characterized in the context of covert communication and covert sensing. Covert communications describe situations in which two legitimate parties attempt to communicate reliably over a noisy channel while avoiding detection by a third party. Covert communications are governed by a square-root law~\cite{Bash2013}, which limits the number of bits that can be reliably and covertly transmitted to the square root of the block length, and the channel-dependent pre-constant that governs the scaling is known for classical discrete-memoryless channels~\cite{Bloch2015b,Wang2016b,Tahmasbi2017}, Gaussian channels~\cite{Wang2016b}, classical-quantum channels~\cite{Sheikholeslami2016,Wang2016c}, and lossy bosonic channels~\cite{Gagatsos2020}. Covert sensing, in contrast, refers to scenarios in which the estimation of parameters of interest requires the use of probing signals that emit energy and are therefore detectable; if estimation could be achieved through purely passive sensing, covertness would automatically be guaranteed. Covert sensing is also governed by a form of square-root law. Specifically,~\cite{Bash2017,Gagatsos2019} have considered the problem of estimating an unknown phase over a bosonic channel while keeping the sensing undetectable by a passive quantum adversary. This operation is made possible by the presence of thermal noise, which allows one to hide the useful sensing signal in the background thermal noise and results in a mean-square phase estimation error scaling as $\bigO{\frac{1}{\sqrt{n}}}$ if $n$ is the number of modes.~\cite{Goeckel2017} has investigated a slightly different model in which the objective is to covertly estimate the impulse response of a linear system. One of the main results obtained is that the bandwidth of sensing signals must scale linearly with the time duration of these signals. Potential applications of covert sensing include covert radar and covert pilot estimation in wireless communications. 

The present work studies covert sensing by drawing on connections with \emph{active hypothesis testing}~\cite{Naghshvar2013a,Naghshvar2013}, also known as \emph{controlled sensing}~\cite{Nitinawarat2013}. Active hypothesis testing differs from traditional hypothesis testing~\cite{Blahut1973}, by considering the possibility of changing the kernel through which unknown parameters are observed, which leads to estimation strategies that are potentially faster or more accurate. Recent studies of active hypothesis testing have built upon the pioneering work of Chernoff~\cite{Chernoff1959} on sequential design of experiments to provide new insights into the problem, including the benefits of sequentiality and adaptivity~\cite{Naghshvar2013a,Nitinawarat2013,Hayashi2009a}, the role of extrinsic Jensen-Shannon Divergence as an information utility metric~\cite{Naghshvar2012}, the unavoidable trade-off between reliability and resolution of estimation~\cite{Naghshvar2013}, and the identification of situations when pure (non-randomized) policies are optimal~\cite{Nitinawarat2013}. Examples of recent applications of active hypothesis testing~ include radar~\cite{Franceschetti2017} and millimeter-wave beam alignment~\cite{Chiu2019}. Active hypothesis testing offers a natural framework for studying covert sensing since covertness effectively requires one to use different observation kernels to hide the presence of probing signals.

The problem of quantum state or channel discrimination without covertness constraint has been intensively studied, see, e.g.,~\cite[Chapter 3]{Watrous_2018},~\cite{Pirandola2019}. The optimal error exponent of discrimination of finitely many quantum states has been characterized~\cite{Nussbaum_Szkola_2009, Li_2016} and resembles the optimal error exponent of \emph{classical} states discrimination; this exponent is known as the multiple quantum Chernoff distance. For \emph{quantum channel} discrimination, as the probing signal can be any quantum state and could be arbitrarily entangled with the environment and previously received signals, several intriguing phenomena, specific to the quantum world, could happen. For example, the legitimate user can substantially decrease the probability of estimation error by keeping its environment entangled with the probing signals~\cite[Example 3.36]{Watrous_2018} or quantum channels that cannot be perfectly distinguished with a single probing can be distinguished with multiple probing with zero probability of error~\cite{Acin_2001}.

In a previous conference paper~\cite{Tahmasbi2020b}, we revisited the idea of covert sensing put forward in~\cite{Bash2017,Gagatsos2019} from the perspective of active hypothesis testing in a classical setting, which we called \emph{active covert sensing}. Therein, we have characterized the exponent of the probability of detection error subject to covertness constraints for non-adaptive non-sequential strategies and illustrated the benefits of adaptive non-sequential strategies. In the present work, we expand upon these results in the quantum setting, in which the legitimate parties attempt to discriminate quantum channels subject to a covertness constraint. Some of the results developed hereafter supersede those in~\cite{Tahmasbi2020b} but, unlike~\cite{Tahmasbi2020b}, we do not consider the adaptation of the probing signals with respect to the previous outputs of the quantum channel; instead, we explore the potential benefits of using entanglement for covert sensing. Specifically, we offer the following two contributions: \begin{inparaenum}[i)]
\item when the legitimate user's probing signals are classical, but the received state by both the legitimate user and the warden are quantum, we characterize the exact detection error exponent of non-adaptive non-sequential strategies subject to a covertness constraint;
\item when the legitimate user's channel is a unitary depending on the unknown parameter, we show that the legitimate user can estimate the unknown parameter \emph{with zero error} while satisfying a stronger notion of covertness compared to what we could achieve over classical-quantum (cq)-channels. We also prove a converse result showing that the asymptotic scaling of the covertness in our achievability result is optimal. 
\end{inparaenum}

The remainder of the paper is organized as follows. We introduce our notation in Section~\ref{sec:notation} and formalize the problem in Section~\ref{sec:problem-formulation}. We provide our main results for cq-channels and unitary channels in Section~\ref{sec:main-res} and prove them in Section~\ref{sec:proofs}. We defer the most technical parts of the proofs to the appendices to streamline the presentation.
%We formally introduce the model for active covert sensing in Section~\ref{sec:problem-formulation}. Fixed-length non-adaptive covert strategies are developed in Section~\ref{sec:fixed-length-non-adaptive} while fixed-length adaptive covert strategies are analyzed in Section~\ref{sec:fixed-length-adapt}. Although the results presented here are mainly of theoretical interest, potential applications include covert target tracking.
% \begin{figure}[h]
%   \centering
%   \scalebox{0.6}{\input{graphics/system-model.tex}}
%   \caption{Model for covert detection}
%   \label{fig:system-model}
% \end{figure}

\section{Notation}
\label{sec:notation}
We denote a vector of length $n$ (e.g., $(x_1, \cdots, x_n)$ in $\calX^n$)  by a boldface letter (e.g. $\mathbf{x}$).  We let $T_{\mathbf{x}}$ denote the type of the vector $\mathbf{x}$, which is a \ac{PMF} over $\calX$ defined by
\begin{align}
T_{\mathbf{x}}(a) \eqdef \frac{\card{\set{i\in\intseq{1}{n}: x_i = a  }}}{n},
\end{align}
where $\intseq{m}{n} \eqdef \set{i\in \mathbb{Z}: m\leq i\leq n}$. We define $\calP_n(\calX) \eqdef \set{T_{\mathbf{x}}: \mathbf{x}\in \calX^n}$ and $\calT_Q \eqdef \set{\mathbf{x}: T_{\mathbf{x}}  = Q}$ for $Q\in \calP_n(\calX)$. $P_X$ denotes a probability distribution over the set $\calX$ and $P_X \otimes P_Y$ is the product distribution over $\calX \times \calY$ induced by two marginals $P_X$ and $P_Y$. $P_X^\pn$ also denotes the $n$-fold product distribution of $P_X$ over $\calX^n$.  We define $\Hb{x} \eqdef -x\log x - (1-x)\log(1-x)$ for $x\in[0, 1]$. Let $\indic{\cdot}$ denote the indicator function. We use standard asymptotic notation $\bigO{\cdot}$, $o(\cdot)$, $\omega(\cdot)$, and $\Theta(\cdot)$. To emphasize that the constant hidden in $\bigO{\cdot}$ could only depend on a parameter $\theta$, we write $\bigO[\theta]{\cdot}$.

A quantum system $A$ is described by a finite-dimensional Hilbert space, which we also denote by $A$ with a slight abuse of notation. Let $\dim A$ be the dimension of $A$ and $\one_A$ be the identity map on $A$. We denote the tensor product of $A$ and $B$ by $A\otimes B$ or $AB$. $\calL(A)$ denotes the set of all  linear operators from $A$ to $A$ and $\calD(A)$ denotes the set of all density operators acting on $A$, which are the possible states of the quantum system $A$. Given two density operators $\rho_A \in \calD(A)$ and $\rho_B\in \calD(B)$, we denote the product state on $AB$ by $\rho_A \otimes \rho_B$. We also define $\rho_A^\pn$ is also $n$-times tensor product of $\rho_A$. A pure state is of the form $\kb{\phi}_A$ for a unit vector $\ket{\phi}_A\in A$. We use $\phi_A$ to denote $\kb{\phi}_A$ when there is no confusion. For $X\in \calL(A)$, the trace norm of $X$ is ${\norm{X}}_1 \eqdef \tr{\smash{\sqrt{X^\dagger X}}}$, and $\nu(X)$ denotes the number of \emph{distinct} eigenvalues of $X$. We also define the support of $X\in \calL(\calH_A)$ as the subspace orthogonal to $\ker X$, which we denote by $\text{supp}(X)$. We denote the adjoint of $X$ by $X^\dagger$. When $X$ is Hermitian, i.e., $X = X^\dagger$, $\lambda_{\min}(X)$ denotes the minimum eigenvalue of $X$. The fidelity between two density operators $\rho$ and $\sigma$ is defined as $F(\rho, \sigma) \eqdef {\norm{\smash{\sqrt{\rho} \sqrt{\sigma}}}}_1^2$. A quantum channel $\calN_{A\to B}$ is a linear trace-preserving completely positive map from $\calL(A)$ to $\calL(B)$. Given two quantum channels $\calN$ and $\calM$, we denote their tensor product by $\calN \otimes \calM$. Let $\id_A$ be the identity channel on $\calL(A)$. For two states $\rho$ and $\sigma$ with $\text{supp}(\rho) \subset \text{supp}(\sigma) $, we define
\begin{align}
 \chi_2\pr{\rho\|\sigma} &\eqdef \tr{\rho^2\sigma^{-1}}-1, \\
 \D{\rho}{\sigma} &\eqdef \tr{\rho \pr{\log \rho - \log \sigma}}.
\end{align}
Additionally, given the spectral decomposition of a state $\sigma = \sum_i \lambda_i P_i$, we define
\begin{align}
    \eta(\rho\|\sigma) = \sum_{i\neq j} \frac{\log \lambda_i - \log \lambda_j}{\lambda_i - \lambda_j} \tr{(\rho - \sigma) P_i (\rho-\sigma) P_j} + \sum_i \frac{1}{\lambda_i} \tr{(\rho - \sigma) P_i (\rho-\sigma) P_i}.
\end{align}
Finally, we  use  standard notions from differential geometry such as tangent space and derivative of a smooth functions.
\section{Problem formulation}
\label{sec:problem-formulation}

\begin{figure}[b]
  \centering
  \scalebox{0.6}{\input{quantum-system-model.tex}}
  \caption{Model for quantum covert sensing.}
  \label{fig:system-model}
\end{figure}
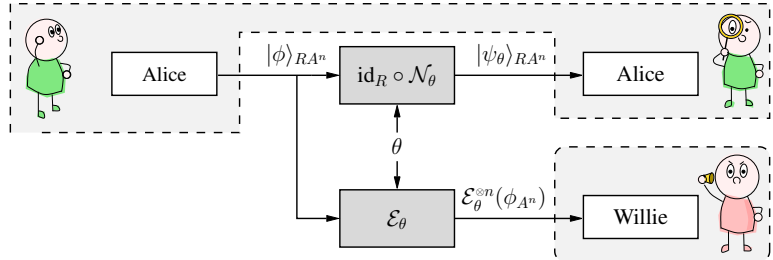

Let $A$, $B$, and $W$ be quantum systems and $\Theta$ be a finite set of parameters. As illustrated in Fig.~\ref{fig:system-model}, let $\set{\calN_\theta : \calL(A) \to \calL(B) }_{\theta\in \Theta}$ and $\set{\calE_\theta : \calL(A) \to \calL(W) }_{\theta\in \Theta}$  be two families of quantum channels. Alice's estimation strategy consists of the following. Alice prepares a possibly entangled state $\ket{\phi}_{RA^n}$, where $R$ is a reference system with $\dim R = \dim A^n$, and transmits the subsystem $A^n$ in the quantum state $\phi_{A^n}$. Alice then receives $\psi_{\theta, RA^n} \eqdef (\id_R \otimes \calN_\theta^\pn) \pr{\kb{\phi}_{RA^n}}$, on which Alice performs a POVM $\set{\Gamma_\theta}_{\theta \in \Theta}$ to estimate the unknown parameter $\theta$. We measure the estimation error through
\begin{align}
\label{eq:def-est-err}
    \max_{\theta\in \Theta} \pr{1 - \tr{\psi_{\theta, RA^n} \Gamma_\theta}}.
\end{align}
Let $\calS(n, \epsilon)$ denote the set of all states $\ket{\phi}_{RA^n}$ such that the estimation error is not greater than $\epsilon$ for \emph{some} POVM. The strategy is called non-sequential because the parameter $n$ is fixed, and non-adaptive because the probing signals are not adapted to past observations.
% We define  for a fixed $R$ and $n$,
% \begin{align}
%     \calS \eqdef \set{\ket{\phi}_{RA^n}: \bra{\phi} (\one_R \otimes (U_\theta^\dagger U_{\theta'})^\pn) \ket{\phi} = 0~\forall \theta\neq \theta' }.
% \end{align}
% If Alice prepares a state $\ket{\phi}_{RA^n} \in \calS$, then she can estimate $\theta$ without error. 

Willie observes what Alice transmits through the memory-less channel $\calE_\theta$ when the parameter is $\theta$, i.e., Willie receives $\calE^{\pn}_\theta(\phi_{A^n})$. One of the input vectors is denoted by $\ket{0} \in A$ and corresponds to Alice being ``inactive,'' i.e., Willie expects Alice to transmit $\ket{0}^\pn$ when no estimation strategy is run. This allows us  measure the inability of Willie to detect probing signals through the covertness metric
\begin{align}
\label{eq:def-cov}
    \max_{\theta\in \Theta}\D{\calE_\theta^{\pn}(\phi_{A^n})}{\calE_\theta^\pn\pr{\kb{0}^\pn}}.
\end{align}
We refer the reader to~\cite{Bloch2015b,Tahmasbi2017} for a discussion on how upper-bounding \eqref{eq:def-cov} yields a bound on the probability of error of any strategy employed by Willie to detect the presence of an estimation strategy. 

We finally define quantities
\begin{align}
C(n, \epsilon) &\eqdef \min_{\ket{\phi}_{RA^n} \in \calS(n, \epsilon) } \max_{\theta\in \Theta}\D{\calE_\theta^{\pn}(\phi_{A^n})}{\calE_\theta^\pn\pr{\kb{0}^\pn}},\\
E(n, \delta) &\eqdef \inf\set{\epsilon \in ]0, 1[: C(n, \epsilon) \leq \delta},
\end{align}
which will be useful to express the fundamental limits of Alice's performance.
\begin{remark}
To ensure that our model is physically realizable, $\calN_\theta$ and $\calE_\theta$ should be consistent for all $\theta$. That is, there should exist  quantum systems $B'_\theta$, $W'_\theta$, and $C_\theta$, isomorphic isometries $V_{\theta, B}: B'_\theta C_\theta \to B$ and $V_{\theta,W}: W'_{\theta}C_{\theta} \to W$ (i.e., $B'_\theta C_\theta \cong B$ and $W'_\theta C_\theta \cong W$) and  a quantum channel $\calM_\theta: \calL(A) \to \calL(B'W'C)$ such that 
\begin{align}
    \calN_\theta(X) &= V_{\theta, B} \ptr{W'_\theta}{\calM_\theta(X)} V_{\theta, B}^\dagger~\quad \forall X\in \calL(A)\\
    \calE_\theta(X) &= V_{\theta, W} \ptr{B'_\theta}{\calM_{\theta}(X)} V_{\theta, W}^\dagger\quad \forall X \in \calL(A).
\end{align}
Because some parts of Alice's and Willie's output systems  could be in common, the order in which Alice and Willie observe their outputs matter. We assume that Willie first observes his outputs and, should he decides to disturb the systems, Alice is notified through another means of communication.
\end{remark}
\section{Main results}
\label{sec:main-res}
\subsection{Cq-channels}
\label{sec:main-res-cq}
We first provide the full asymptotic characterization of Alice's optimal performance for covert sensing over cq-channels. Although the cq channels we consider here have finite dimension, cq channels are good models for those channels that arise in quantum optics, such as bosonic channels in which the input is a classical parameter of the transmitted states, see, e.g.,~\cite{Bash2017}. In particular, we assume that both $\calN_\theta$ and $\calE_\theta$ are cq for all $\theta$. That is, there exist an  orthonormal basis  $\set{\ket{a_u}_A: u \in \calU}$  for $A$ and two sets of quantum states $\set{\rho_{B|\theta}^u}_{u \in \calU, \theta\in \Theta}$ and $\set{\rho_{W|\theta}^u}_{u \in\calU, \theta\in \Theta}$   such that 
\begin{align}
    \calN_\theta\pr{\ketbra{a_{u}}{a_{u'}}} &= \indic{u = u'}\rho_{B|\theta}^u \quad \forall u, u' \in \calU\\
    \calE_\theta\pr{\ketbra{a_{u}}{a_{u'}}} &= \indic{u = u'}\rho_{W|\theta}^u \quad \forall u, u' \in \calU.
\end{align}
Note that $\calE_\theta$ and $\calN_\theta$ are uniquely characterized by the linearity of quantum channels.  We also assume that $0\in \calU$ and that $\ket{a_0}$ is the innocent state $\ket{0}$.
We also impose some mild restrictions on the cq-channels that are required to make the problem meaningful.
\begin{enumerate}
    \item $\widetilde{\Theta} \eqdef \set{\theta \in \Theta: \exists \theta' \in \Theta\setminus\set{\theta}: \rho_{B|\theta}^0 = \rho_{B|\theta'}^0  } \neq \emptyset$. This ensures that Alice cannot distinguish all parameters by sending only $0$, which would result in perfect covertness. 
    \item There exists $\theta \in \Theta$ such that   no distribution $P$ over $\calU\setminus \set{0}$ is such that $\sum_u P(u) \rho_{W|\theta}^u  = \rho_{W|\theta}^0$. 
    This ensures that Alice cannot simulate sending $0$ by a random selection of other inputs.
    \item For all $\theta\in \Theta$ and for all $u \in \calU$ $\textnormal{supp}\pr{\rho_{W|\theta}^u} \subset \textnormal{supp}\pr{\rho_{W|\theta}^0}$. This ensures that Willie cannot detect the estimation with non-vanishing probability when a state with support not included in $\textnormal{supp}\pr{\rho_{W|\theta}^0}$ is transmitted.
    \end{enumerate}
We introduce the notion of conditional Chernoff information to state our main result.
\begin{definition}
Let $\theta$ and $\theta'$ be two parameters in $\Theta$ and $P$ be a distribution over $\calU$. The conditional Chernoff information  is
\begin{align}
\label{eq:chernoff-inf-def}
\Dcc{\theta}{\theta'}{P} \eqdef \sup_{s\in [0, 1]} - \sum_u P(u) \log\pr{\tr{\pr{\rho_{B|\theta}^u}^s \pr{\rho_{B|\theta'}^{u}}^{1-s}}}.
\end{align}
\end{definition}
\begin{theorem}
\label{th:n-adap-disc}
Under the assumption on the cq-channel discussed above, we have $-\log E(n, \delta_n) = \Theta(\sqrt{n\delta_n})$ for any sequence $\set{\delta_n}_{n\geq 1} = \bigO{1} \cap \omega\pr{\frac{\log n}{n}}$, and in particular,
\begin{align}
\label{eq:non-adaptive}
    \lim_{n\to \infty} -\frac{\log E(n, \delta_n)}{\sqrt{n\delta_n}}  = \sup_{P} \frac{\sqrt{2}\min_{\theta\neq \theta': \rho_{B|\theta}^u = \rho_{B|\theta}^0}\Dcc{\theta}{\theta'}{P}}{\sqrt{ \max_{\theta } \eta(\sum_u P(u)\rho_{W|\theta}^u\| \rho_{W|\theta}^0)}},
\end{align}
where the supremum is taken over all probability distributions $P$ over $\calU \setminus \set{0}$.
\end{theorem}
\begin{remark}
\cite[Th.~1]{Tahmasbi2020b} can be obtained as a special case of Theorem~\ref{th:n-adap-disc}, corresponding to the situation in which all operators in $\set{\qy{u}}_{u\in \calU, \theta \in \Theta}$ and all operators in $\set{\qz{u}}_{u\in \calU, \theta \in \Theta}$ mutually commute.
\end{remark}
\subsection{Unitary channels: the power of entanglement}
We now consider a situation in which $\set{U_\theta}_{\theta\in \Theta}$ is a family of unitaries acting on $A$ such that $\calN_\theta(\rho) = U_\theta \rho U_\theta^\dagger$ for all $\theta \in \Theta$. This corresponds to an ideal situation in which Alice is able to receive all transmitted signals without any loss. We also assume that $\calE_\theta$ is independent of $\theta$, i.e., there exists a quantum channel $\calE:\calL(A)\to \calL(W)$ such that $\calE_\theta = \calE$ for all $\theta\in \Theta$. The latter assumption helps us  simplify the expression of our results but our proof does not exploit this assumption. Also note that this assumption does not trivialize the problem since we still require the sensing to be covert.

\subsubsection{Achievability}
We require again mild assumptions in our achievability result to make the problem meaningful.
\begin{enumerate}
    \item $U_\theta \neq U_{\theta'}$ for all $\theta\neq \theta'$. Without this assumption Alice would be unable to distinguish at least two parameters.
    \item $\supp{\calE(\rho)} \subset \supp{\calE(\kb{0})}$ for all $\rho \in \calD(A)$. Without this assumption, the transmission of such $\rho$ would allow Willie to systematically detect Alice.
\end{enumerate}
\begin{theorem}
\label{th:entang-cov-est-part1}
Under the above assumptions, there exists a positive integer $N$ depending on $\set{U_\theta}_{\theta\in \Theta}$ such that for all $n \geq N$
\begin{align}
    C(n, 0) \leq \bigO[\set{U_\theta}, \calE]{\frac{1}{n}}.
\end{align}
\end{theorem}
Note that there is a significant difference between the optimal performance of  unitary channels and cq-channels. Indeed, according to Theorem~\ref{th:n-adap-disc}, we have $C(n, \exp(-\bigO{\sqrt{n} \delta}) \geq \delta $ for all $\delta>0$ and for all cq-channels, while $C(n, 0) \leq \bigO{1/n}$ when Alice's channel is a unitary for all parameters $\theta$. As we show next, the rate of decay of the covertness metric with $n$ is optimal under mild assumptions. 
\subsubsection{Converse}
Our converse result holds under the following mildly restrictive assumptions.
\begin{enumerate}
    \item $U_\theta \ket{0} = U_{\theta'} \ket{0}$ for some $\theta \neq \theta'$, i.e., Alice cannot distinguish all parameters by always sending  $\ket{0}$ and trivially ensuring covertness.
    \item $\calE(\rho) \neq \calE(\kb{0})$ for all $\rho \in \calD(A) \setminus \set{\kb{0}}$, i.e., Alice cannot simulate sending $\ket{0}$ using other quantum states.
    \item  There exists no sequence $\set{\rho_n}_{n\geq 1} \subset \calD(A) \setminus{\kb{0}}$ such that 
\begin{align}
    \lim_{n\to \infty} \frac{\norm[1]{\rho_n - \kb{0}}}{\norm[1]{\calE(\rho_n) - \calE(\kb{0})}} = \infty. 
\end{align}
This last assumption prevents Alice to send states whose image under $\calE$ is close to $\calE(\kb{0})$. We show that testing this assumption is possible by providing a \emph{computable} equivalent form in Lemma~\ref{lm:cont-q-ch}.
\end{enumerate}

\begin{theorem}
\label{th:entang-cov-est-part2}
    Under the above assumptions, for all $\epsilon \in [0, 1]$, we have
    \begin{align}
        C(n, \epsilon) \geq \bigO[\calE ]{\frac{ \pr{1-\epsilon}^4}{n}}.
    \end{align}
\end{theorem}
%This result shows that asymptotic  scaling of our achievability result in Theorem~\ref{th:entang-cov-est-part1} is optimal under some assumptions on the quantum channels.
\begin{remark}
 Note that Theorem~\ref{th:entang-cov-est-part1} and Theorem~\ref{th:entang-cov-est-part2} require different assumptions because we need to restrict Willie's power for achievability and Alice's power for converse.
\end{remark}

We now provide a computable equivalent form of our last assumption in the converse result. We first need the following definition, which introduces a map from $\calL(A)$ to a Euclidean space.
\begin{definition}
\label{def:f}
Let $d\eqdef \dim A $ and $\set{\ket{a_1}, \cdots, \ket{a_d}}$ be an orthonormal basis for $A$ such that $\ket{a_1} = \ket{0}$. We define a function $f:\calL(A) \to \R^{2d^2}$ which maps $X\in \calL(A)$ to the vector $(\textnormal{Re}(\bra{a_i} X \ket{a_j}), \textnormal{Im}(\bra{a_i} X \ket{a_j}))_{i,j\in \intseq{1}{d}}$.

We also define $2d - 2$ vectors  $a_1, \cdots, a_{2d-2} \in \R^{2d^2}$ such that the $j$th component of $a_i$ is $\indic{i = j} + \indic{j = \lceil i / 2 \rceil d + 1}$.
\end{definition}
%$f(\ker \calE) \cap \textnormal{span}\set{a_1, \cdots, a_{2d-2}} = \set{0}$
We show in Appendix~\ref{sec:cont-q-ch} that the vectors $a_1, \cdots, a_{2d-2} \in \R^{2d^2}$ form a basis for the tangent space of 
\begin{align}
    f(\set{\kb{\phi}: \phi \in A, \nrm{\phi} = 1} \setminus \kb{0})
\end{align}
at the origin.  

\begin{lemma}
\label{lm:cont-q-ch}
Let $\calE: \calL(A) \to \calL(B)$ be a quantum channel and $\ket{0} \in A$ be a unit vector. Suppose that $\calE(\rho) \neq \calE(\kb{0})$ for all $\rho \in \calD(A) \setminus \set{\kb{0}}$. 
We then have 
\begin{align}
\sup_{\rho  \in \calD(A) \setminus \set{\kb{0}}} \frac{\nrm{\rho - \kb{0}}_1}{\nrm{\calE(\rho) - \calE(\kb{0})}_1} < \infty
\end{align}
if and only if $f(\ker \calE) \cap \textnormal{span}(a_1, \cdots, a_{2d - 2}) = \set{0} $ where the function $f$ and the vectors $a_1, \cdots, a_{2d - 2}$ are defined in Definition~\ref{def:f}.
\end{lemma}

\begin{proof}
See Appendix~\ref{sec:cont-q-ch}.
\end{proof}

\begin{remark}
If $\calE:\calL(A) \to \calL(E)$ is an invertible map (as a linear map), there exists no sequence satisfying the above conditions, which is consistent with our result as $\ker \calE = \set{0}$. However, there might be some quantum channels $\calE$ that  are not invertible, but for which we still have $f(\ker \calE) \cap \textnormal{span}\set{a_1, \cdots, a_{2d-2}} = \set{0}$.
\end{remark}
\section{Proofs}
\label{sec:proofs}
\subsection{Achievability proof of Theorem~\ref{th:n-adap-disc}}
We first derive a general bound on the reliability of a strategy when the input is generated according to $P_{\mathbf{U}}$.
\begin{lemma}
\label{lm:one-shot-achv}
Let $P_{\mathbf{U}}$ be any distribution over $\calU^n$. There exists an $(n, \epsilon, \delta)$ non-adaptive strategy with
\begin{align}
\label{eq:one-shot-est}
    \log \epsilon &= \max_{\theta\neq \theta'}\log \pr{\sum_{Q \in \calP_{n}(\calU)}P_{\mathbf{U}}(\calT_{Q}) \exp\pr{-n\Dcc{\theta}{\theta'}{Q}}} +  \bigO[\dim B, \card{\calU}, \card{\Theta}]{\log n}\\
    \delta &= \max_{\theta\in \Theta} \D{\sum_{\mathbf{u}} P_{\mathbf{U}}(\mathbf{u}) \qpz{\mathbf{u}}}{\qpz{\mathbf{0}}}.
    \label{eq:one-shot-cov}
\end{align}
\end{lemma}

\begin{proof}
See Appendix~\ref{sec:one-shot-achv}.
\end{proof}

Deterministic strategies, for which $P_{\mathbf{U}}$ is positive only on one input sequence,  cannot achieve any positive exponent as shown next. Let $\widetilde{\Theta}$ be as defined in Section~\ref{sec:main-res-cq} and $\theta\in \widetilde{\Theta}$ be such that   no distribution $P$ over $\calU\setminus \set{0}$ is such that $\sum_u P(u) \qz{u}= \qz{0}$. If Alice transmits a fixed sequence $\mathbf{u}$, we have
\begin{align}
\delta 
&\geq \D{\qpz{\mathbf{u}}}{\qpz{\mathbf{0}} } \\
& \geq \card{\set{i\in \intseq{1}{n}: u_i \neq  0 }} \min_{u \in \calU \setminus \set{0}} \D{\qz{u}}{\qz{0}}.
\end{align}
By our assumption on $\theta$,  $\min_{u \in \calU \setminus \set{0}} \D{\qz{u}}{\qz{0}}$ is positive. Therefore, the number of non-zero elements of $\mathbf{u}$ is uniformly bounded. By  definition of $\widetilde{\Theta}$, there also exists $\theta' \in \Theta \setminus \set{\theta}$ such that $\qy{0} = \qyp{0}$. Thus, even when restricting the parameter set to $\set{\theta, \theta'}$, the estimation error cannot vanish. Hence, no positive exponent is achievable.

Furthermore, \ac{iid} actions cannot achieve the optimal exponent since, with exponentially small probability, the type of the input sequence largely deviates from the typical  input type and affects the achievable exponent.

We now introduce an input probability distribution $P_{\mathbf{U}}$ that circumvents the challenges discussed above. Intuitively, $P_{\mathbf{U}}$ should be such that 1) we can control the type of the sequences in its support and 2) we can ensure covertness.  Let $P$ be any distribution over $\calU$ and define $\alpha \eqdef 1 - P(0)$. We  set for $\zeta > 0$,
\begin{align}
\label{eq:pun-def1}
    \calQ &\eqdef \set{Q \in \calP_n(\calU): |Q(u) - P(u)| \leq \alpha \zeta~\forall u \in \calU\setminus\set{0}},\\
    \calA &\eqdef \cup_{Q \in \calQ} \calT_Q,\\ 
    P_{\mathbf{U}}(\mathbf{u}) &\eqdef \begin{cases}\frac{P^\pn(\mathbf{u})}{P^\pn( \calA)}\quad& \mathbf{u} \in \calA, \\ 0 \quad & \mathbf{u} \notin \calA.  \end{cases}\label{eq:pun-def}
\end{align}
Intuitively, the parameter $\alpha$ allows us to finely control the type of sequences in $\calA$ with $\alpha$ possibly depending on $n$.
In the following lemma, we provide bounds on \eqref{eq:one-shot-est} and \eqref{eq:one-shot-cov} for this specific choice of $P_{\mathbf{U}}$.
\begin{lemma}
\label{lm:spec-pu-prop}
Let $\theta$ and $\theta'$ be two distinct elements of $\Theta$. We have
\begin{multline}
    \log\pr{ \sum_{Q \in \calP_{n}(\calU)}P_{\mathbf{U}}(\calT_{Q}) \exp\pr{-n\Dcc{\theta}{\theta'}{Q}}} \leq \min\left(-n \Dcc{\theta}{\theta'}{P} - \bigO[\set{\qy{u}}]{n\alpha\zeta \card{\calU}}, \right.\\
    \left.n(1- \alpha(1 + \zeta \card{\calU})) \inf_{s\in [0, 1]} \log \pr{ \tr{ \pr{\qy{0}}^s \pr{\qyp{0}}^{1-s}}}\right).\label{eq:est-err-1}
\end{multline}
In addition, we have
\begin{multline}
\label{eq:covertness-pu}
    \D{\sum_{\mathbf{u}} P_{\mathbf{U}}(\mathbf{u}) \qpz{\mathbf{u}}}{\qpz{\mathbf{0}}} 
    \leq n \D{\sum_u P(u) \qz{u}}{\qz{0}} \\
    + 2 \card{\calU}\exp\pr{-\frac{\alpha n \zeta^2}{3}} \log \pr{\frac{\dim W}{\lambda_{\min}\pr{\qz{0}}}} n + \Hb{2 \card{\calU}\exp\pr{-\frac{\alpha n \zeta^2}{3}}}.
\end{multline}
\end{lemma}
\begin{proof}
See Appendix~\ref{sec:spec-pu-prop}.
\end{proof}
We are now ready to prove the achievability of the exponent in \eqref{eq:non-adaptive}. Let $\overline{P}$ be any distribution over $\calU\setminus \set{0}$ (not depending on $n$) and $\set{\lambda_n}_{n\geq 0}$ be a vanishing sequence specified later. We define
\begin{align}
   \label{eq:alpha-def} \alpha_n & \eqdef \sqrt{\frac{2\delta_n(1-\lambda_n) }{n\max_{\theta} \eta(\sum_u \overline{P}(u) \qz{u}\| \qz{0})}},\\
    P(u) &\eqdef \begin{cases}1 - \alpha_n \quad &u = 0,\\ \alpha_n \overline{P}(u)\quad &u \neq 0.\end{cases}
\end{align}
We then choose $P_{\mathbf{U}}$ according to \eqref{eq:pun-def}, for which we have
\begin{align}
   & \D{\sum_{\mathbf{u}} P_{\mathbf{U}}(\mathbf{u}) \qpz{\bfu}}{\qpz{\mathbf{0}}} \\
   & \stackrel{(a)}{\leq} n \D{\sum_u P(u) \qz{u}}{\qz{0}} + 2 \card{\calU}\exp\pr{-\frac{\alpha_n n \zeta^2}{3}} \log \frac{ \dim W}{\lambda_{\min}\pr{\qz{0}}} n + \Hb{2 \card{\calU}\exp\pr{-\frac{\alpha_n n \zeta^2}{3}}}\\
    &\stackrel{(b)}{=} n \D{\sum_u P(u) \qz{u}}{\qz{0}}+ \exp\pr{-\omega\pr{\log n}}\\
    &\stackrel{(c)}{=} n\pr{\frac{\alpha_n^2}{2}\eta(\sum_u \overline{P}(u)\qz{u}\| \qz{0} ) + \bigO{\alpha_n^3} } +  \exp\pr{-\omega\pr{\log n}} \\
    &\stackrel{(d)}{\leq} (1-\lambda_n)\delta_n + \bigO{\alpha_n^3 n}  +  \exp\pr{-\omega\pr{\log n}}\\
    &= (1-\lambda_n)\delta_n + \bigO{\frac{\delta_n^{\frac{3}{2}}}{\sqrt{n}}}  +  \exp\pr{-\omega\pr{\log n}}
\end{align}
where $(a)$ follows from \eqref{eq:covertness-pu}, $(b)$ follows since we are choosing $\delta_n = \omega\pr{\frac{\log n}{n}}$, $(c)$ follows from Lemma~\ref{lm:approx-rel-entropy} in Appendix~\ref{sec:approx-rel-entropy} and $(d)$ follows from \eqref{eq:alpha-def}. We set
\begin{align}
    \lambda_n = \bigO{\frac{\delta_n^{\frac{1}{2}}}{\sqrt{n}}}  + \frac{1}{\delta_n} \exp\pr{-\omega\pr{\log n}},
\end{align}
which is vanishing and ensures that $\D{\sum_{\mathbf{u}} P_{\mathbf{U}}(\mathbf{u}) \qpz{\bfu}}{\qpz{\mathbf{0}}} \leq \delta_n$. Therefore, by Lemma~\ref{lm:one-shot-achv}, there exists an $(n, \epsilon_n, \delta_n)$ with
\begin{align}
    \log \epsilon_n = \max_{\theta\neq \theta'} \log\pr{ \sum_{Q \in \calP_{n}(\calU)}P_{\mathbf{U}}(\calT_{Q}) \exp\pr{-n\Dcc{\theta}{\theta'}{Q}}} + \bigO[\dim B, \card{\calU}, \card{\Theta}]{\log n} \label{eq:est-error}.
\end{align}
To  upper-bound $\epsilon_n$, we consider two cases for $\theta$ and $\theta'$. If $\qy{0} = \qyp{0}$, then \eqref{eq:est-err-1} yields that
\begin{align}
-\log \pr{\sum_{Q\in \calP_n(\calU)}P_{\mathbf{U}}(\calT_{Q}) \exp\pr{-n \Dcc{\theta}{\theta'}{Q}}}
&\geq n\Dcc{\theta}{\theta'}{P} +\bigO[\set{\qz{u}}]{n\alpha_n\zeta \card{\calU}} \\
&\stackrel{(a)}{\geq} n\alpha_n(1 + \zeta \card{\calU} \bigO[\set{\qz{u}}]{1})) \Dcc{\theta}{\theta'}{\overline{P}}, \label{eq:bound-zero-eq-epsilon}
\end{align}
where $(a)$ follows from the definition of $\overline{P}$ and since all terms in the definition of $\Dcc{\theta}{\theta'}{P}$ are non-negative. If $\qy{0} \neq \qyp{0}$,  we have
\begin{multline}
-\log \pr{P_{\mathbf{U}}(\calT_{Q}) \exp\pr{-n\Dcc{\theta}{\theta'}{P}}}\\  \geq  -(1- \alpha_n(1 + \zeta \card{\calU})) n \inf_{s\in [0, 1]} \log \pr{\tr{ \pr{\qy{0}}^s \pr{\qyp{0}}^{1-s}}} \stackrel{(a)}{=} \Theta(n),
\end{multline}
where $(a)$ follows since $\inf_{s\in [0, 1]} \log \pr{\tr{ \pr{\qy{0}}^s \pr{\qyp{0}}^{1-s}}} < 0$ when $\qy{0} \neq \qyp{0}$.
Therefore, we can exclude all pairs $(\theta, \theta')$ with $\qy{0} \neq \qyp{0}$ from the maximization in \eqref{eq:est-error} for large enough $n$. Thus, using Lemma~\ref{lm:one-shot-achv} and \eqref{eq:bound-zero-eq-epsilon}, we have
\begin{align}
-\log \epsilon_n \geq n\alpha_n  (1 + \zeta \card{\calU}\bigO[\set{\qz{u}}]{1}) \min_{\theta, \theta': \qy{0} = \qyp{0}}  \Dcc{\theta}{\theta'}{\overline{P}}  + \bigO[\dim B, \card{\calU}, \card{\Theta}]{\log n}.
\end{align}
Using the definition of $\alpha_n$ in \eqref{eq:alpha-def}, we obtain
\begin{align}
\liminf_{n\to \infty} - \frac{\log \epsilon_n}{\sqrt{\delta_n n}} \geq \frac{\sqrt{2}\min_{\theta\neq \theta': \qy{0} = \qyp{0}}\Dcc{\theta}{\theta'}{\overline{P}}}{\sqrt{ \max_{\theta} \eta(\sum_u \overline{P}(u)\qz{u}\| \qz{0})}}.
\end{align}

\subsection{Converse proof of Theorem~\ref{th:n-adap-disc}}
Let us consider a sequence of $(n, \epsilon_n, \delta_n)$ non-adaptive strategies, for which the input is generated according to a \ac{PMF} $P_{\bfU}$ over $\calU^n$ in the $n^{th}$ strategy. We define
\begin{align}
    \overline{P} &\eqdef \frac{1}{n} \sum_{i=1}^n P_{U_i},\\
    \alpha_n &\eqdef 1- \overline{P}(0),\\
    \widetilde{P}(u) &\eqdef \begin{cases}\frac{\overline{P}(u)}{\alpha_n}\quad& u \neq 0,\\ 0 \quad & u = 0.\end{cases}
\end{align}

\begin{lemma}
\label{lm:one-shot-converse}
We have 
\begin{align}
    \label{eq:one-shot-converse1}
    -\log \epsilon_n \leq   \min_{\theta\neq \theta':\qy{0} = \qyp{0}}  n \alpha_n \Dcc{\theta}{\theta'}{\widetilde{P}} - \bigO[\dim B, \card{\calU}, \card{\Theta}]{\frac{\log n}{n}}.
\end{align}
As $n$ goes to infinity, $\alpha_n$ tends to zero and we have
\begin{align}
    \label{eq:one-shot-converse2}
    \frac{\delta_n}{n} \geq \frac{\alpha_n^2}{2} \max_{\theta}\eta(\sum_u \widetilde{P}(u) \qz{u}\| \qz{0}) + \bigO[\set{\qz{0}}]{\alpha_n^3}.
\end{align}
\end{lemma}
\begin{proof}
See Appendix~\ref{sec:one-shot-converse}.
\end{proof}
We therefore have
\begin{align}
    -\frac{\log \epsilon_n}{\sqrt{\delta_n n}}
    &\stackrel{(a)}{\leq} \frac{  n\alpha_n \min_{\theta\neq \theta': \qy{0} = \qyp{0}}\Dcc{\theta}{\theta'}{\widetilde{P}}- \bigO[\dim B, \card{\calU}, \card{\Theta}]{\frac{\log n}{n}}}{\sqrt{\delta_n n}}\\
    &\stackrel{(b)}{\leq} \frac{ \sqrt{\frac{2n\delta_n(1+o(1))}{\max_{\theta}\eta(\sum_u \widetilde{P}(u) \qz{u}\| \qz{0})}} \min_{\theta\neq \theta': \qy{0} = \qyp{0}}\Dcc{\theta}{\theta'}{\widetilde{P}}}{\sqrt{\delta_n n}}\\
    &= \sqrt{\frac{2(1+o(1))}{\max_{\theta}\eta(\sum_u \widetilde{P}(u) \qz{u}\| \qz{0})}} \min_{\theta\neq \theta': \qy{0} = \qyp{0}}\Dcc{\theta}{\theta'}{\widetilde{P}}
\end{align}
where $(a)$ follows from \eqref{eq:one-shot-converse1}, and $(b)$ follows from \eqref{eq:one-shot-converse2} and the constraint $\delta_n = \omega(\log n/n)$. Taking the limit  as $n$ goes to infinity, we obtain the desired converse bound.

 \subsection{Proof of Theorem~\ref{th:entang-cov-est-part1}}
We first recall from~\cite{Acin_2001} that given a unitary $U\neq \one_A$ acting on $A$, one can find a positive integer $m$ and a pure state $\ket{\nu}_{A^m}$ such that $\bra{\nu} U^{\proddist m} \ket{\nu} = 0$. Applying this result to the unitary $U_\theta^\dagger U_{\theta'} \neq \one_A$, there exist a  positive integer $m_{\theta, \theta'}$ and pure state $\ket{\nu_{\theta, \theta'}}_{A^{m_{\theta, \theta'}}}$ in $A^{m_{\theta, \theta'}}$ such that $\bra{\nu_{\theta, \theta'}} (U_\theta^\dagger U_{\theta'})^{\proddist m_{\theta, \theta'}} \ket{\nu_{\theta, \theta'}} = 0$. Let $m \eqdef \sum_{\theta\neq \theta'} m_{\theta, \theta'}$ and $\ket{\nu}_{A^m}$ be a pure state in $A^m$ defined as the tensor product of all $\ket{\nu_{\theta, \theta'}}_{A^{m_{\theta, \theta'}}}$ in an arbitrary order. Let $\ell \eqdef \lfloor n / m \rfloor$. Alice decomposes the first $m\ell$ channel uses into $\ell$ sub-blocks of length $m$, selects one sub-block at random, transmits $\ket{\nu}_{A^m}$ on that sub-block, and transmits $\ket{0}$ for any other channel use. By transmitting $\ket{\nu}_{A^m}$, Alice can estimate $\theta$ without error. 

We now analyze the covertness. Let us denote the state transmitted through the channels by 
\begin{align}
    \phi_{A^n} \eqdef \frac{1}{\ell} \sum_{i=1}^\ell (\kb{0})^{\proddist(i-1)m} \otimes \kb{\nu}_{A^m}\otimes (\kb{0})^{n - im}. 
\end{align}
Note that
\begin{align}
    \D{\calE^\pn(\phi_{A^n})}{\calE^\pn(\kb{0}^\pn)} 
    &\stackrel{(a)}{\leq} \frac{\chi_2(\calE^{\proddist m}(\kb{\nu}_{A^m}) \| \calE^{\proddist m}(\kb{0}^{\proddist m}))}{\ell}\\
    & \leq \frac{1}{\ell \lambda_{\min}(\calE(\kb{0}))^m}\\
    &\stackrel{(b)}{\leq} \frac{m}{(n-m) \lambda_{\min}(\calE(\kb{0}))^m},
\end{align}
where $(a)$ follows from~\cite[Eq. (B144)]{tahmasbi2018framework}, and $(b)$ follows since $\ell>(n-m)/m$. Since $m$ is a constant independent of $n$, we obtain the desired bound on the covertness.

\subsection{Proof of Theorem~\ref{th:entang-cov-est-part2}}
We consider a general strategy, in which Alice initially prepares $\ket{\phi}_{RA^n}$ and, after receiving $\ket{\psi_\theta}_{RA^n} \eqdef \pr{\one_R \otimes U_\theta^\pn} $ for an unknown parameter $\theta$, performs a POVM to estimate $\theta$. We assume that the the estimation error as defined in \eqref{eq:def-est-err} is $\epsilon$ and the covertness as defined in \eqref{eq:def-cov} is $\delta$. We desire to prove that $\delta \geq \bigO[\calE]{(1-\epsilon)^4/n}$. We show this result in three steps sketched as follows.
\begin{enumerate}
    \item We first use the assumption that $U_\theta \ket{0} = U_{\theta'}\ket{0}$ for some $\theta \neq \theta'$ to show that $\epsilon \geq 1 - 2\nrm{\phi_{A^n} - \pr{\kb{0}}^\pn}_1$.
    \item We upper-bound $\nrm{\phi_{A^n} - \pr{\kb{0}}^\pn}_1$ by $\bigO[\calE]{\pr{n\sum_{i=1}^n \D{\calE(\phi_{A_i})}{\calE(\kb{0})}}^{\frac{1}{4}}}$. The proof of this step relies on both  our assumptions on $\calE$, i.e., $\calE(\rho) \neq \calE(\kb{0})$ for all $\rho \in \calD(A) \setminus \set{\kb{0}}$ and 
    \begin{align}
    \sup_{\rho  \in \calD(A) \setminus \set{\kb{0}}} \frac{\nrm{\rho - \kb{0}}_1}{\nrm{\calE(\rho) - \calE(\kb{0})}_1} < \infty.
    \end{align}
    \item We use standard converse argument for covert communication to show that 
    \begin{align}
    \sum_{i=1}^n \D{\calE(\phi_{A_i})}{\calE(\kb{0})} \allowbreak \leq \delta.
    \end{align}
\end{enumerate}
Combining these three steps yields that $\epsilon \geq 1 - \bigO[\calE]{(n\delta)^{\frac{1}{4}}}$, which is equivalent to $\delta \geq \bigO[\calE]{(1-\epsilon)^{4}/n}$ as desired. We now prove each step.

\paragraph{Proof of step 1}  The estimation error, $\epsilon$, is lower-bounded by 
\begin{align}
    \max_{\theta \neq \theta'}\abs{\braket{\psi_\theta}{\psi_{\theta'}}}^2 
    &= \max_{\theta \neq \theta'} \abs{\bra{\phi} \pr{\one_R \otimes (U_\theta^\dagger U_{\theta'})^\pn}\ket{\phi}}^2\\
    &= \max_{\theta \neq \theta'} \abs{\tr{\phi_{A^n} (U_\theta^\dagger U_{\theta'})^\pn}}^2\\
    &= \max_{\theta \neq \theta'} \abs{\tr{\pr{ \pr{\kb{0}}^\pn - \phi_{A^n} } (U_\theta^\dagger U_{\theta'})^\pn} - \tr{\pr{\kb{0}}^\pn (U_\theta^\dagger U_{\theta'})^\pn}}^2\\
    &\stackrel{(a)}{\geq}  \min_{\theta \neq \theta'} \abs{\tr{\pr{\phi_{A^n} - \pr{\kb{0}}^\pn} (U_\theta^\dagger U_{\theta'})^\pn} - 1}^2\\
    &\stackrel{(b)}{\geq}  \min_{\theta \neq \theta'} 1 - 2\abs{\tr{\pr{\phi_{A^n} - \pr{\kb{0}}^\pn} (U_\theta^\dagger U_{\theta'})^\pn}}\\
    &\stackrel{(c)}{\geq} 1- 2\nrm{\phi_{A^n} - \pr{\kb{0}}^\pn}_1\\
\end{align}
where $(a)$ follows because $U_\theta \ket{0} = U_{\theta'} \ket{0} $ for some $\theta \neq \theta'$, $(b)$ follows from $|1-z|^2 = 1-2\text{Re}(z) + |z|^2 \geq 1- 2|z|$ for any complex number $z$, and $(c)$ follows from $\abs{\tr{XY}}\leq \nrm{X}\nrm{Y}_1$ for all $X, Y\in \calL(A^n)$ and $\nrm{(U_\theta^\dagger U_{\theta'})^\pn} = 1$. 

\paragraph{Proof of step 2} We have 
\begin{align}
    \nrm{\phi_{A^n} - \pr{\kb{0}}^\pn}_1
    &\stackrel{(a)}{\leq} \sqrt{1 - F(\phi_{A^n}, \pr{\kb{0}}^\pn)}\\
    &=\sqrt{1 - \bra{0}^\pn \phi_{A^n} \ket{0}^\pn}\\
    &\stackrel{(b)}{\leq} \sqrt{\sum_{i=1}^n\pr{1- \bra{0} \phi_{A_i} \ket{0}} } \\
    &\stackrel{(c)}{\leq }\sqrt{\sum_{i=1}^n\nrm{ \kb{0}  - \phi_{A_i} }_1 },
\end{align}
where $(a)$ follows from~\cite[Theorem 1]{Fuchs_1999}, $(b)$ follows from (classical) union bound, and $(c)$ follows since $1-F(\rho, \sigma) \leq \nrm{\rho-\sigma}_1$ when $\rho$ is pure. We now state a lemma that allows us to bound $\nrm{\rho - \kb{0}}_1$ using $\nrm{\calE(\rho - \kb{0})}_1$
By our assumption on $\calE$ there exists $B>0$ such that  for all $i\in \intseq{1}{n}$,
\begin{align}
    \nrm{ \kb{0}  - \phi_{A_i} }_1 
    &\leq B \nrm{\calE(\kb{0}) - \calE(\phi_{A_i})}_1\\
    &\leq B \sqrt{\D{\calE(\phi_{A_i})}{\calE(\kb{0})}}.
\end{align}
This implies that 
\begin{align}
\nrm{\phi_{A^n} - \pr{\kb{0}}^\pn}_1 
&\leq \sqrt{B} \sqrt{\sum_{i=1}^n \sqrt{\D{\calE(\phi_{A_i})}{\calE(\kb{0})}}}\\
&\leq \sqrt{B}  \sqrt{\sqrt{n} \sqrt{\sum_{i=1}^n \D{\calE(\phi_{A_i})}{\calE(\kb{0})}}}
\end{align}
% \begin{align}
%     \abs{\bra{\phi} \pr{\one_R \otimes (U_\theta^\dagger U_{\theta'})^\pn} \ket{\phi} - \bra{0}^\pn(U_\theta^\dagger U_{\theta'})^\pn \ket{0}^\pn  } \leq \nrm{\phi_{A^n} - \kb{0}^n}_1
% \end{align}
\paragraph{Proof of step 3} We have
\begin{align}
    \D{\calE^{\pn}(\phi_{A^n})}{\calE^\pn\pr{\kb{0}^\pn}} 
    &= -H(\calE^{\pn}(\phi_{A^n})) + \tr{\calE^\pn\pr{\phi_{A^n}} \log \pr{\pr{\calE(\kb{0})}^\pn} }\\
    &=  -H(\calE^{\pn}(\phi_{A^n})) +  \sum_{i=1}^n \tr{\calE(\phi_{A_i}) \log\pr{\calE(\kb{0})} }\\
    &\geq -\sum_{i=1}^n H(\calE(\phi_{A_i})) +  \sum_{i=1}^n \tr{\calE(\phi_{A_i}) \log\pr{\calE(\kb{0})} }\\
    &= \sum_{i=1}^n \D{\calE(\phi_{A_i})}{\calE(\kb{0})}.
\end{align}

\appendices
\section{Approximation of quantum relative entropy}
We characterize  in the next lemma the expansion of $\D{\alpha \rho_1  + (1-\alpha)\rho_0}{\rho_0}$ in $\alpha$ around zero.
\label{sec:approx-rel-entropy}
\begin{lemma}
\label{lm:approx-rel-entropy}
Let $\rho_1$ and $\rho_0$ be two density operators on $A$ such that $\rho_0$ is invertible. We have for small $\alpha > 0$
\begin{align}
    \D{\alpha \rho_1  + (1-\alpha)\rho_0}{\rho_0} = \frac{1}{2} \alpha^2 \eta(\rho_1\|\rho_0) + \bigO[\rho_0]{\alpha^3}.
\end{align}
\end{lemma}
\begin{remark}
 This result is similar to~\cite[Lemma~1]{Wang2016c}, but the expression for $\eta(\rho_1\|\rho_0)$ in~\cite{Wang2016c} is 
\begin{align}
    \tr{\int_0^\infty \rho_1(\rho_0 + s )^{-1} \rho_1(\rho_0 + s)^{-1}ds} - 1,
\end{align}
which involves an integration. In addition, in Lemma~\ref{lm:approx-rel-entropy}, the constant behind the higher order term is independent of $\rho_1$, which is not shown in~\cite[Lemma~1]{Wang2016c} and which  is crucial in our converse argument. 
\end{remark}
We first recall two results from functional calculus before proving Lemma~\ref{lm:approx-rel-entropy}.
\begin{lemma}
\label{lm:log-appx}
Let $X$ be a positive operator in $\calL(\calH)$ with eigen-decomposition $X = \sum_{i=1}^q \lambda_i P_i$, where $\lambda_1, \cdots, \lambda_q$ are distinct eigenvalue of $X$ and $P_i$ is the projection onto the eigen-subspace corresponding to $\lambda_i$. There exists $\epsilon > 0$ such that for all $Y$ with $\norm{X-Y}\leq \epsilon$, $X+Y$ is positive and 
\begin{align}
    \log(X+Y) = \log(X) + \sum_{i,j} D_{i,j} P_iYP_j + \bigO[X]{\norm{Y}^2},
\end{align}
where 
\begin{align}
    D_{i,j} = \begin{cases}
    \frac{\log \lambda_i - \log \lambda_j}{\lambda_i - \lambda_j}&\quad i\neq j\\
    \frac{1}{\lambda_i}&\quad i=j
    \end{cases}
\end{align}
\end{lemma}
 \begin{proof}
 
 It follows from applying ~\cite[Th.~4.2]{shapiro2002differentiability} to the function $\log(\cdot)$.
 \end{proof}
\begin{lemma}
\label{lm:diff-tr}
Let $I\subset \R$ be an open interval in $\R$ and  $f:I\to \R$ be a smooth function. Let $A$ and $B$ be two self-adjoint operator in $\calL(\calH)$. We define $g(t) \eqdef \tr{f(A+tB)}$ for all $t$ such that all eigenvalues of $A+tB$ are in $I$. Then, the domain of $g$ is open and for each $t$ in the domain of $g$,
\begin{align}
    g'(t) = \tr{f'(A+tB) B}.
\end{align}
\end{lemma}
\begin{proof}
See~\cite[Eq.~(11.176)]{Wilde_2017}.
\end{proof}
We now prove Lemma~\ref{lm:approx-rel-entropy}. Let $\rho_0$ has eigen-decomposition $\sum_i \lambda_i P_i$ and define
\begin{align}
    D_{i,j} = \begin{cases}
    \frac{\log \lambda_i - \log \lambda_j}{\lambda_i - \lambda_j}&\quad i\neq j\\
    \frac{1}{\lambda_i}&\quad i=j.
    \end{cases}
\end{align}
We also define $\Delta \eqdef \rho_1 - \rho_0$ and 
\begin{align}
    g(\alpha) 
    &\eqdef \D{\alpha \rho_1 + (1-\alpha) \rho_0}{\rho_0}\\
    &= \D{\rho_0 + \alpha \Delta}{\rho_0} \\
    &=\tr{\pr{\rho_0 + \alpha \Delta}\log\pr{\rho_0 + \alpha \Delta}} -\tr{\pr{\rho_0 + \alpha \Delta}\log\pr{\rho_0}}.
\end{align}
Note that 
\begin{align}
    g'(\alpha) 
    &\stackrel{(a)}{=} \tr{\Delta(\log(\rho_0+\alpha \Delta) + \one )} - \tr{\Delta \log \rho_0}\\
    &= \tr{\Delta\pr{\log(\rho_0+\alpha \Delta) - \log(\rho_0)}}\\
    &\stackrel{(b)}{=} \tr{\Delta\pr{ \sum_{i,j}D_{i,j}P_i(\alpha\Delta)P_j  + \bigO[\rho_0]{\norm{\alpha \Delta}^2} } }\\
    &= \tr{\Delta\pr{ \sum_{i,j}D_{i,j}P_i(\alpha\Delta)P_j  + \bigO[\rho_0]{\norm{\alpha \Delta}^2} }} \\
    &= \alpha \sum_{i,j} D_{i,j}\tr{\Delta P_i \Delta P_j}  + \tr{\Delta\bigO[\rho_0]{\norm{\alpha\Delta}^2}}\\
    &\stackrel{(c)}{=} \alpha \sum_{i,j} D_{i,j}\tr{\Delta P_i \Delta P_j}  + \bigO[\rho_0]{\alpha^2}\\
    &= \alpha \eta(\rho_1\|\rho_0) + \bigO[\rho_0]{\alpha^2}\label{eq:f-prime-rel}
\end{align}
where $(a)$ follows from Lemma~\ref{lm:diff-tr}, $(b)$ follows from Lemma~\ref{lm:log-appx}, and $(c)$ follows since the norm of $\Delta = \rho_1-\rho_0$ is bounded for all density operators $\rho_0$ and $\rho_1$.
We then have
\begin{align}
    \abs{g(\alpha) - \frac{1}{2}\alpha^2 \eta(\rho_1\|\rho_0) }
    &\stackrel{(a)}{=} \abs{\int_0^\alpha \pr{g'(\beta) - \beta \eta(\rho_1\|\rho_0) }d\beta}\\
    &\leq \int_0^\alpha \abs{g'(\beta) - \beta \eta(\rho_1\|\rho_0)} d\beta \\
    &\stackrel{(b)}{=} \int_0^\alpha \bigO[\rho_0]{\beta^2} d\beta\\
    &= \bigO[\rho_0]{\alpha^3}
\end{align}
where $(a)$ follows from the fundamental theorem of calculus and $(b)$ follows from Eq.~\eqref{eq:f-prime-rel}.

\section{Proof of Lemma~\ref{lm:one-shot-achv}}
\label{sec:one-shot-achv}
We first recall a result from~\cite{Li_2016} on the optimal performance of discriminating multiple quantum states.
\begin{lemma}
\label{lm:mult-hyp-q}
Let $\set{\rho_\theta}_{\theta\in \Theta}$ be a finite family of density operators acting on a finite dimensional space. There exists a \ac{POVM} $\set{\Gamma_\theta}_{\theta\in \Theta}$ such that 
\begin{align}
    \max_{\theta\in \Theta} \tr{\rho_\theta \pr{\one - \Gamma_\theta}} \leq 10 (\card{\Theta}-1)^2 \max_{\theta\in \Theta} \nu\pr{\rho_\theta} \sum_{\theta\neq \theta'} \inf_{s\in  [0, 1]}\tr{\rho_\theta^s \rho_{\theta'}^{1-s}}
\end{align}
\end{lemma}
\begin{proof}
It follows from combining~\cite[Th.~2]{Li_2016} and~\cite[Eq.~(35)]{Li_2016}.
\end{proof}
Alice samples the input sequence $\mathbf{u}$ according to $P_{\mathbf{U}}$ and receives $\rho_{\mathbf{B}|\theta}^{\mathbf{u}} \eqdef \rho_{B|\theta}^{u_1} \otimes \cdots \otimes \rho_{B|\theta}^{u_n}$. Alice then performs the POVM $\set{\Gamma^{\mathbf{u}}_{\theta}: \theta \in \Theta}$ given by Lemma~\ref{lm:mult-hyp-q} for the states $\set{\qpy{\mathbf{u}}}_{\theta\in \Theta}$, resulting in estimation error 
\begin{align}
    10 (\card{\Theta}-1)^2 \max_{\theta\in \Theta} \nu\pr{\qpy{\mathbf{u}}}  \sum_{\theta\neq \theta'} \inf_{s\in  [0, 1]}\tr{\pr{\qpy{\mathbf{u}}}^s \pr{\qpyp{\mathbf{u}}}^{1-s}}.
\end{align}
Note that 
\begin{align}
    \max_{\theta\in \Theta} \nu\pr{\qpy{\mathbf{u}}} \leq (n+1)^{\dim B \card{\calU}},
\end{align}
and
\begin{align}
    \inf_{s\in  [0, 1]}\tr{\pr{\qpy{\mathbf{u}}}^s \pr{\qpyp{\mathbf{u}}}^{1-s}}
    &=  \inf_{s\in  [0, 1]}\prod_{i=1}^n \tr{\pr{\qy{{u_i}}}^s \pr{\qyp{{u_i}}}^{1-s}}\\
    &= \exp\pr{-\sup_{s\in [0, 1]} - \sum_{i=1}^n \log\pr{ \tr{\pr{\qy{{u_i}}}^s \pr{\qyp{{u_i}}}^{1-s}}} }\\
    &= \exp\pr{-n \sup_{s\in [0, 1]} - \sum_{u\in \calU}T_{\mathbf{u}}(u) \log \pr{\tr{\pr{\qy{{u}}}^s \pr{\qyp{{u}}}^{1-s}} }}\\
    &= \exp\pr{-n \Dcc{\theta}{\theta'}{T_{\mathbf{u}}}}.
\end{align}
This concludes the proof.

% We want to show that for $\alpha$ small
% \begin{align}
%     f(\alpha) = \frac{1}{2} \alpha^2 f''(0) + O(\alpha^3),
% \end{align}
% where the constant hidden in $O(\cdot)$ depends only on $\rho_0$ an
\section{Proof of Lemma~\ref{lm:spec-pu-prop}}
\label{sec:spec-pu-prop}
We introduce a notation that simplifies our expressions. Let us define for $u\in \calU$ and $s\in[0, 1]$,
\begin{align}
f(s;u) \eqdef \log\pr{\tr{\pr{\qz{u}}^s \pr{\qz{u}}^{1-s}}},
\end{align}
which is always non-positive and
\begin{align}
    L_f \eqdef \min_{u'\in\calU \setminus \set{0}} \min_{s'\in[0, 1]} f(s';u') > -\infty.
\end{align}
We then have
\begin{align}
    \log\pr{\sum_{Q \in \calP_{n}(\calU)}P_{\mathbf{U}}(\calT_{Q}) \exp\pr{-n\Dcc{\theta}{\theta'}{Q}}}\label{eq:up-b-log-beg}
    &\leq \max_{Q\in\calP_{n}(\calU): P_{\mathbf{U}}(\calT_Q) > 0  }\left[ -n\Dcc{\theta}{\theta'}{Q}\right]\\
    &= \max_{Q\in\calP_{n}(\calU): P_{\mathbf{U}}(\calT_Q) > 0  }\left[ -n\sup_{s\in [0, 1]}-\sum_{u\in \calU}Q(u)f(s;u) \right]\\
    &\stackrel{(a)}{\leq}  \max_{Q\in\calP_{n}(\calU): P_{\mathbf{U}}(\calT_Q) > 0  }\left[ -n\sup_{s\in [0, 1]}-\sum_{u\in \calU \setminus \set{0}}Q(u)f(s;u) \right]\\
    &\stackrel{(b)}{\leq}  -n\sup_{s\in [0, 1]}-\sum_{u\in \calU \setminus \set{0}}\pr{P(u) - \zeta \alpha}f(s;u)\\
    &\leq \Dcc{\theta}{\theta'}{P} - n\alpha\zeta \card{\calU} L_f,\label{eq:up-b-log-end}
\end{align}
where $(a)$ follows since $f(s;0) \leq 0$, and $(b)$ follows since $P_{\mathbf{U}}(\calT_Q) > 0$ if and only if $|P(u) - Q(u)| \leq \alpha\zeta $ for all $u\in \calU\setminus\set{0}$ (see Eq.~\eqref{eq:pun-def1}-\eqref{eq:pun-def}).
Note also that $Q(0) = 1 - \sum_{u\in \calU\setminus\set{0}} Q(u) \geq 1 - \sum_{u\in \calU\setminus\set{0}} (P(u) + \alpha \zeta  ) \geq P(0) - \alpha \zeta \card{\calU} = 1 - \alpha(1+\zeta\card{\calU})$ for all $Q \in \calQ$. The same line of reasoning as in Eq.~\eqref{eq:up-b-log-beg}-\eqref{eq:up-b-log-end} then provides that
\begin{align}
   \log\pr{\sum_{Q \in \calP_{n}(\calU)}P_{\mathbf{U}}(\calT_{Q}) \exp\pr{-n\Dcc{\theta}{\theta'}{Q}}}
    &\leq -n(1 - \alpha(1+\zeta\card{\calU})) \inf_{s\in[0, 1]} f(s, 0),
\end{align}
which yields Eq.~\eqref{eq:est-err-1} together with Eq.~\eqref{eq:up-b-log-end}.
% Therefore, if $P_{\mathbf{U}}(\calT_Q) > 0$, we have for all $s\in ]0, 1[$
% \begin{align}
%     &\sum_u Q(u){\log \pr{\sum_y \py{u}(y)^s \pyp{u}(y)^{1-s}}}\displaybreak[0]\\
%     &\stackrel{(a)}{\leq}  \sum_{u\in \calU \setminus \set{0}} (P(u) - \alpha \zeta){\log \pr{\sum_y \py{u}(y)^s \pyp{u}(y)^{1-s}}}\displaybreak[0]\\
%     &\leq -\alpha\zeta \card{\calU} \min_{u\in \calU \setminus \set{0}} \log \pr{\sum_y \py{u}(y)^s \pyp{u}(y)^{1-s}} +\sum_{u\in \calU \setminus \set{0}} P(u) {\log \pr{\sum_y \py{u}(y)^s \pyp{u}(y)^{1-s}}},\displaybreak[0]
% \end{align}
% where $(a)$ follows since $\log \pr{\sum_y \py{u}(y)^s \pyp{u}(y)^{1-s}} \leq 0$.

Let $\mathbf{U}$ be distributed according to $P^\pn$ and $T_{\mathbf{U}}$ denote its type, which is a random element of $\calP_n(\calU)$. We have
\begin{align}
1- P^\pn(\calA) 
&\stackrel{(a)}{\leq} \sum_{u\in \calU \setminus \set{0}} \P{ |T_{\mathbf{U}}(u) - P(u)| \geq \zeta \alpha }\\
&\stackrel{(b)}{=} \sum_{u\in \calU \setminus \set{0}: P(u) > 0} \P{ |T_{\mathbf{U}}(u) - P(u)| \geq \zeta \alpha }\\
&\stackrel{(c)}{\leq} \sum_{u\in \calU \setminus \set{0}: P(u) > 0}  2\exp\pr{-\frac{\alpha^2 n \zeta^2}{3P(u)}}\\
&\stackrel{(d)}{\leq}  2 \card{\calU}\exp\pr{-\frac{\alpha n \zeta^2}{3}},
\end{align} 
where $(a)$ follows from the union bound, $(b)$ follows since $T_{\mathbf{U}}(u) = 0$ with probability one when $P(u) = 0$, $(c)$ follows from a Chernoff bound, and $(d)$ follows since $P(u) \leq 1 - P(0) = \alpha$ for all $u\in \calU\setminus \set{0}$. Note that $\frac{1}{2}{\norm{P_{\mathbf{U}} - P^\pn}}_1 = 1- P^\pn(\calA) \leq  2 \card{\calU}\exp\pr{-\frac{\alpha n \zeta^2}{3}}$ by the definition of $P_{\mathbf{U}}$. Hence, the data processing inequality implies that

\begin{align}
\frac{1}{2}{\norm{\sum_{\mathbf{u}}P_{\mathbf{U}}(\mathbf{u})\qpz{\mathbf{u}} - \sum_{\mathbf{u}}P^\pn(\mathbf{u})\qpz{\mathbf{u}}}}_1 \leq  2 \card{\calU}\exp\pr{-\frac{\alpha n \zeta^2}{3}}.
\end{align}

Finally, the following  continuity result  for the relative entropy completes the proof of \eqref{eq:covertness-pu}. 
\begin{lemma}
\label{lm:rel-cont}
Let $\rho_{\mathbf{B}}$ and $\sigma_{\bfB}$ be two density operators over $B^n$ such that $\frac{1}{2}\nrm{\rho_{\bfB} - \sigma_{\bfB}}_1 \leq \epsilon$.  We then have
\begin{align}
\abs{\D{\rho_{\bfB}}{\qpz{\mathbf{0}}} - \D{\sigma_{\bfB}}{\qpz{\mathbf{0}}}} \leq  \epsilon \log\pr{\frac{\dim B}{\pr{\lambda_{\min}(\qz{0})}^2}}n  + \Hb{\epsilon}.
\end{align}
\end{lemma}
\begin{proof}
Note that
\begin{align}
\displaybreak[0]\abs{\D{\rho_{\mathbf{B}}}{\qpz{\mathbf{0}}} - \D{\sigma_{\mathbf{B}}}{\qpz{\mathbf{0}}}}
&= \abs{H(\rho_{\mathbf{B}}) - H(\sigma_{\mathbf{B}}) + \tr{\pr{\rho_{\bfB} - \sigma_{\bfB}}\log\pr{\qpz{\mathbf{0}}} }}\\
\displaybreak[0]&\leq \abs{H(\rho_{\mathbf{B}}) - H(\sigma_{\mathbf{B}}) } +  \abs{\tr{\pr{\rho_{\bfB} - \sigma_{\bfB}}\log\pr{\qpz{\mathbf{0}}} }}\\
\displaybreak[0]&\stackrel{(a)}{\leq} \epsilon \log(\dim B ) n + \Hb{\epsilon}+\abs{\tr{\pr{\rho_{\bfB} - \sigma_{\bfB}}\log\pr{\qpz{\mathbf{0}}} }}\\
\displaybreak[0]&\leq \epsilon \log( \dim B) n + \Hb{\epsilon}+ \nrm{\rho_{\bfB} - \sigma_{\bfB}}_1 \nrm{\log \pr{\qpz{\mathbf{0}}}} \\
&\leq \epsilon \log( \dim B) n + \Hb{\epsilon}+ 2\epsilon  \log \pr{\frac{1}{\lambda_{\min}(\qz{0})}} n
\end{align}
where $(a)$ follows from  Fannes' inequality.
\end{proof}
\section{Proof of Lemma~\ref{lm:one-shot-converse}}
\label{sec:one-shot-converse}
Eq.~\eqref{eq:one-shot-converse2} follows from the same argument used to obtain~\cite[Eq.~(39)]{Wang2016c}, except using Lemma~\ref{lm:approx-rel-entropy} instead of~\cite[Lemma~1]{Wang2016c}.
We prove Eq.~\eqref{eq:one-shot-converse1} in four steps as summarized below.
\begin{itemize}
    \item \textbf{Step 1:} We lower-bound the estimation error of this strategy by
    \begin{align}
    \max_{\theta \in \Theta} \sum_{\bfu} P_{\bfU}(\bfu) \tr{\Gamma^{\bfu}_\theta\qpy{\bfu}} \geq \frac{1}{\card{\Theta}}\max_{\theta \neq \theta'} \sum_{\bfu}P_{\bfU}(\bfu) \pr{1 - \frac{1}{2}\nrm{\qpy{\bfu} - \qpyp{\bfu}}_1}.
    \end{align}
    \item \textbf{Step 2:} Let us now consider the spectral decomposition of $\qy{u} = \sum_{y\in \calY} p_{\theta}^u(y) \kb{e_{\theta}^u(y)}$, where $\calY$ is a set of size $\dim B$, $p_{\theta}^u$ is a \ac{PMF} over $\calY$, and $\set{\ket{e_{\theta}^u(y)}: y\in \calY}$ forms an orthonormal basis for $B$. We also define $q_{\theta, \theta'}^u(y, y') \eqdef p_{\theta}^u(y)\abs{\braket{e_{\theta}^u(y)}{e_{\theta'}^u(y')}}^2 $, which is a \ac{PMF} over $\calY \times \calY$, and $q_{\theta}^{\bfu} \eqdef q_{\theta}^{u_1} \otimes \cdots \otimes q_{\theta}^{u_n}$, which is a \ac{PMF} over $(\calY \times \calY)^n$. We shall show that
    \begin{align}
    1 - \frac{1}{2}\nrm{\qpy{\bfu} - \qpyp{\bfu}}_1 \geq \frac{1}{2}\pr{ 1 - \frac{1}{2}\nrm{q_{\theta, \theta'}^{\bfu} - q_{\theta', \theta}^{\bfu}}_1}.
\end{align}
\item \textbf{Step 3:} Let $P_{V|U}$ and $P_{\widetilde{V}|U }$ be two conditional distributions, $\bfu \in \calU^n$ be a sequence with type $T_U$, and $\bfV$ and $\widetilde{\bfV}$ be distributed according to $P_{\bfV} = P_{V|U = u_1} \otimes \cdots \otimes P_{V|U=u_n}$ and $P_{\widetilde{\bfV}} = P_{\widetilde{V}|U = u_1} \otimes \cdots \otimes P_{\widetilde{V}|U=u_n}$, respectively. We shall show that
\begin{multline}
    \sum_{\bfv } \min\pr{P_{\bfV}(\bfv), P_{\widetilde{\bfV}}(\bfv)} \\\geq \exp\pr{-n\sup_{s\in [0, 1]} \log\pr{\sum_u T_U(u) \sum_v P_{V|U}(v|u)^s P_{\widetilde{V}|U}(v|u)^{1-s}} + \bigO{\frac{\log n}{n}}}.
\end{multline}
\item \textbf{Step 4:} We show that 
    \begin{align}
        1 - \frac{1}{2}\nrm{\qpy{\bfu} - \qpyp{\bfu}}_1  
        &= \exp\pr{-n\Dcc{\theta}{\theta}{T_U} + \bigO{\frac{\log n}{n}}},
    \end{align}
    which concludes the proof together with Jensen's inequality and the convexity of the exponential function.
\end{itemize}
We now provide the detailed proof of each step.
\paragraph{Proof of step 1}
Note that 
\begin{align}
    \max_{\theta \in \Theta} \sum_{\bfu} P_{\bfU}(\bfu) \tr{\Gamma^{\bfu}_\theta\qpy{\bfu}}
    &\geq  \frac{1}{\card{\Theta}} \sum_{\bfu}P_{\bfU}(\bfu) \sum_{\theta}  \tr{\Gamma^{\bfu}_\theta\qpy{\bfu}}\\
    &\stackrel{(a)}{\geq}  \frac{1}{\card{\Theta}} \sum_{\bfu}P_{\bfU}(\bfu) \max_{\theta \neq \theta'}\pr{1 - \frac{1}{2}\nrm{\qpy{\bfu} - \qpyp{\bfu}}_1} \\
    &\geq \frac{1}{\card{\Theta}}\max_{\theta \neq \theta'} \sum_{\bfu}P_{\bfU}(\bfu) \pr{1 - \frac{1}{2}\nrm{\qpy{\bfu} - \qpyp{\bfu}}_1},
\end{align}
where $(a)$ follows from the varitional characterization of the trace norm $\frac{1}{2} \norm[1]{\rho - \sigma} = \max_{0 \prec \Gamma \prec \one} \tr{\Gamma (\rho - \sigma)}$.

\paragraph{Proof of step 2}
The proof is in~\cite{Nussbaum_Szkola_2009}, but we provide the proof for completeness. We first define $p_\theta^{\bfu} \eqdef p_\theta^{u_1}\otimes \cdots \otimes p_{\theta}^{u_n}$ and $\ket{e_\theta^{\bfu}(\bfy)} \eqdef \ket{e_\theta^{u_1}(y_1)} \otimes \cdots \otimes \ket{e_\theta^{u_n}(y_n)}$. We have
\begin{align}
    1 - \frac{1}{2}\nrm{\qpy{\bfu} - \qpyp{\bfu}}_1  
    &= \inf_{0 \prec \Gamma \prec \one} \left[\tr{\Gamma\qpy{\bfu}} + \tr{(\one - \Gamma) \qpyp{\bfu}}\right]\\
    &=  \inf_{0 \prec \Gamma \prec \one} \left[ \sum_{\bfy} p_{\theta}^{\bfu}(\bfy) \bra{e_{\theta}^{\bfu}(\bfy)} \Gamma \ket{e_{\theta}^{\bfu}(\bfy)} +\sum_{\bfy} p_{\theta'}^{\bfu}(\bfy) \bra{e_{\theta'}^{\bfu}(\bfy)} (\one - \Gamma) \ket{e_{\theta'}^{\bfu}(\bfy)}  \right]\\
    &\stackrel{(a)}{\geq}  \inf_{0 \prec \Gamma \prec \one} \left[ \sum_{\bfy} p_{\theta}^{\bfu}(\bfy) \nrm{\Gamma \ket{e_{\theta}^{\bfu}(\bfy)}}^2_2 + \sum_{\bfy} p_{\theta'}^{\bfu}(\bfy) \nrm{ (\one - \Gamma) \ket{e_{\theta'}^{\bfu}(\bfy)}}_2^2  \right]\\
    &=  \inf_{0 \prec \Gamma \prec \one} \left[ \sum_{\bfy, \bfy'} p_{\theta}^{\bfu}(\bfy) \abs{\bra{e_{\theta'}^{\bfu}(\bfy')} \Gamma \ket{e_{\theta}^{\bfu}(\bfy)}}^2  + \sum_{\bfy, \bfy'} p_{\theta'}^{\bfu}(\bfy)  \abs{\bra{e_{\theta'}^{\bfu}(\bfy)} (\one - \Gamma) \ket{e_{\theta}^{\bfu}(\bfy')}}^2 \right]\\
    &\geq \inf_{0 \prec \Gamma \prec \one} \sum_{\bfy, \bfy'} \min(p_{\theta}^{\bfu}(\bfy), p_{\theta'}^{\bfu}(\bfy')) \pr{ \abs{\bra{e_{\theta'}^{\bfu}(\bfy')} \Gamma \ket{e_{\theta}^{\bfu}(\bfy)}}^2  +   \abs{\bra{e_{\theta'}^{\bfu}(\bfy')} (\one - \Gamma) \ket{e_{\theta}^{\bfu}(\bfy)}}^2 }\\
    &\stackrel{(b)}{\geq} \frac{1}{2}  \sum_{\bfy, \bfy'} \min(p_{\theta}^{\bfu}(\bfy), p_{\theta'}^{\bfu}(\bfy')) \abs{\braket{e_\theta^{\bfu}(\bfy)}{e_{\theta'}^{\bfu}(\bfy')}}^2\displaybreak[0],
    %&= \frac{1}{2}\sum_{\bfy, \bfy'}\min(q_{\theta, \theta'}^\bfu(\bfy), q_{\theta', \theta}^\bfu(\bfy'))
\end{align}
where $(a)$ follows since $(1-\Gamma) \succ (1-\Gamma)^2$ for all $0\prec \Gamma \prec \one$, and $(b)$ follows from $|x|^2+|y|^2 \geq |x+y|^2/2$ for any two complex numbers $x$ and $y$.

\paragraph{Proof of Step 3}
Deploying standard method of type arguments, we have
\begin{align}
  &\sum_{\bfv } \min\pr{P_{\bfV}(\bfv), P_{\widetilde{\bfV}}(\bfv)} \\
  &\geq \max_{T_{V|U}\in \calP_n(\calV|\bfu)} \sum_{\bfv \in \calT_{T_{V|U}}(\bfu) } \min\pr{P_{\bfV}(\bfv), P_{\widetilde{\bfV}}(\bfv)} \\
  &\stackrel{(a)}{\geq} \max_{T_{V|U}\in \calP_n(\calV|\bfu)} \pr{n+1}^{-\card{\calV}\card{\calU}}\exp\pr{-n\max\pr{\D{T_{V|U}}{P_{V|U}|T_U}, \D{T_{V|U}}{P_{\widetilde{V}|U}|T_U}}}\\
  &= \pr{n+1}^{-\card{\calV}\card{\calU}}\exp\pr{-n \min_{T_{V|U}\in \calP_n(\calV|\bfu)}\max\pr{\D{T_{V|U}}{P_{V|U}|T_U}, \D{T_{V|U}}{P_{\widetilde{V}|U}|T_U}}},
\end{align}
where $(a)$ follows from~\cite[Eq. (2.8)]{Csiszar}. Next note that for an arbitrary conditional distribution $Q_{V|U}$, there exists $T_{V|U}\in \calP_n(\calV|\bfu)$ such that $\Delta_u \eqdef \frac{1}{2}\nrm{T_{V|U = u} - Q_{V|U=u}}_1 \leq \frac{\card{\calV}}{nT_U(u)}$ for all $u\in \supp{T_U}$. Thus, for such a $T_{V|U}$,
\begin{align}
    &\abs{\D{T_{V|U}}{P_{\widetilde{V}|U}|T_U} - \D{Q_{V|U}}{P_{\widetilde{V}|U}|T_U}}\\
    &\leq   \sum_{u} T_U(u)  \abs{\D{T_{V|U=u}}{P_{\widetilde{V}|U=u}} - \D{Q_{V|U=u}}{P_{\widetilde{V}|U=u}}}\\
    &\leq   \sum_{u} T_U(u)\pr{\abs{H(T_{V|U=u}) - H(Q_{V|U=u})} + \sum_v \abs{T_{V|U}(v|u) - Q_{v|u}(v|u)}\log \frac{1}{P_{V|U}(v|u)}} \\
    &\stackrel{(a)}{\leq} \sum_{u} T_U(u)\pr{\Delta_u \log \card{\calV} + \Hb{\Delta_u} + \max_{v}\log \frac{1}{P_{V|U}(v|u)} \Delta_u} \\
     &\stackrel{(b)}{\leq} \sum_{u} T_U(u)\pr{\Delta_u \log \card{\calV} + \Delta_u \log \frac{e}{\Delta_u} + \max_{v}\log \frac{1}{P_{V|U}(v|u)} \Delta_u}\\
    &\stackrel{(c)}{\leq }\sum_{u} T_U(u)\pr{\frac{\card{\calV}}{nT_U(u)} \log \card{\calV} + \frac{\card{\calV}}{nT_U(u)} \log \frac{nT_U(u) e}{\card{\calV}} + \max_{v}\log \frac{1}{P_{V|U}(v|u)} \frac{\card{\calV}}{nT_U(u)}}\\
    &\leq \frac{\card{\calV} \card{\calU} \log \card{\calV} }{n}+\frac{\card{\calV} \card{\calU} \log n }{n} + \max_{v, n} \log \frac{1}{P_{V|U}} \frac{\card{\calV}}{n} = \bigO{\frac{\log n}{n}},
    \end{align}
where $(a)$ follows from Fannes' inequality, $(b)$ follows since $\Hb{x} \leq x \log \frac{e}{x}$, and $(c)$ follows since $\Delta_u\leq \frac{\card{\calV}}{nT_U(u)}$ by our choice of $T_{V|U}$.
Hence,
\begin{multline}
    \min_{T_{V|U}\in \calP_n(\calV|\bfu)}\max\pr{\D{T_{V|U}}{P_{V|U}|T_U}, \D{T_{V|U}}{P_{\widetilde{V}|U}|T_U}}\\ \geq \min_{Q_{V|U}\in \calP(\calV|\calU)}\max\pr{\D{Q_{V|U}}{P_{V|U}|T_U}, \D{Q_{V|U}}{P_{\widetilde{V}|U}|T_U}} - \bigO{\frac{\log n}{n} }.
\end{multline}
Finally,~\cite[Eq.~(39)]{Nitinawarat2013} implies that 
\begin{multline}
    \min_{Q_{V|U}\in \calP(\calV|\calU)}\max\pr{\D{Q_{V|U}}{P_{V|U}|T_U}, \D{Q_{V|U}}{P_{\widetilde{V}|U}|T_U}} \\= - \sup_{s\in [0, 1]} \log\pr{\sum_u T_U(u) \sum_v P_{V|U}(v|u)^s P_{\widetilde{V}|U}(v|u)^{1-s}}.
\end{multline}

\paragraph{Proof of Step 4} Combining the result of step one and two, we have
\begin{multline}
    1 - \frac{1}{2}\nrm{\qpy{\bfu} - \qpyp{\bfu}}_1 \geq \exp\left(-n\sup_{s\in [0, 1]} \log\left(\sum_u T_U(u)\right.\right. \\
    \left.\left.\times \sum_{y, y'} \left(p_\theta^u(y)\abs{\braket{e_\theta^u(y)}{e_{\theta'}^u(y')}}^2\right)^s \pr{p_{\theta'}^u(y') \abs{\braket{e_\theta^u(y)}{e_{\theta'}^u(y')}}^2}^{1-s}\right) + \bigO{\frac{\log n}{n}}\right).\label{eq:l1-decom}
\end{multline}
Note that 
\begin{align}
    &\sum_{y, y'} (p_\theta^u(y) \abs{\braket{e_\theta^u(y)}{e_{\theta'}^u(y')}}^2)^s \pr{p_{\theta'}^u(y') \abs{\braket{e_\theta^u(y)}{e_{\theta'}^u(y')}}^2}^{1-s}\\
    &~~~~~~~~~~~= \sum_{y, y'} p_\theta^u(y)^s p_{\theta}^u(y')^{1-s} \abs{\braket{e_\theta^u(y)}{e_{\theta'}^u(y')}}^2\\
    &~~~~~~~~~~~= \tr{\pr{\sum_y p_\theta^u(y)^s \kb{e_\theta^u(y)}}\pr{\sum_{y'} p_{\theta'}^u(y')^{1-s} \kb{e_{\theta'}^u(y')}}}\\
    &~~~~~~~~~~~= \tr{\pr{\qy{\theta}}^s \pr{\qy{\theta'}}^{1-s}}.\label{eq:tr-cher}
\end{align}
Substituting Eq.~\eqref{eq:tr-cher} into Eq.~\ref{eq:l1-decom}, we have
\begin{align}
    1 - \frac{1}{2}\nrm{\qpy{\bfu} - \qpyp{\bfu}}_1  
    &\geq \exp\pr{-n\sup_{s\in [0, 1]} \log\pr{\sum_u T_U(u) \tr{\pr{\qy{\theta}}^s \pr{\qy{\theta'}}^{1-s}}} + \bigO{\frac{\log n}{n}}}\\
    &= \exp\pr{-n\Dcc{\theta}{\theta}{T_U} + \bigO{\frac{\log n}{n}}},
\end{align}
as desired.

\section{Proof of Lemma~\ref{lm:cont-q-ch}}
\label{sec:cont-q-ch}
We divide the proof into four steps.
 \paragraph{Step 1} By the observation $\nrm{X}_2 \leq  \nrm{X}_1  \leq \sqrt{\dim A} \nrm{X}_2$ for all $X\in \calL(A)$, it holds that
\begin{align}
\sup_{\rho  \in \calD(A) \setminus \set{\kb{0}}} \frac{\nrm{\rho - \kb{0}}_1}{\nrm{\calE(\rho) - \calE(\kb{0})}_1} < \infty
\end{align}
if and only if
\begin{align}
\label{eq:bound-norm-2}
\sup_{\rho  \in \calD(A) \setminus \set{\kb{0}}} \frac{\nrm{\rho - \kb{0}}_2}{\nrm{\calE(\rho) - \calE(\kb{0})}_2} < \infty.
\end{align}
 \paragraph{Step 2} We state a result that relates the norm of the output of a linear operator to the norm of the output of projection onto the kernel of the linear operator. This implies that one only needs to know $\ker \calE$  to verify \eqref{eq:bound-norm-2}.
 \begin{proposition}
\label{prop:bound-operator}
Let $V$ and $W$ be Hilbert spaces and $A:V\to W$ be a non-zero linear map. Let $P$ be the projection onto $\ker A$. There exist positive constants $B_1$ and $B_2$ such that for all $v\in V$,
\begin{align}
    B_1\nrm{A v} \leq \nrm{(\one_V - P) v} \leq B_2 \nrm{A v}. 
\end{align}
\end{proposition}
\begin{proof}
See Appendix~\ref{sec:bound-operator}
\end{proof}
 
By  Proposition~\ref{prop:bound-operator}, we have
\begin{align}
\sup_{\rho  \in \calD(A) \setminus \set{\kb{0}}} \frac{\nrm{\rho - \kb{0}}_2}{\nrm{\calE(\rho) - \calE(\kb{0})}_2} < \infty
\end{align}
if and only if
\begin{align}
\sup_{\rho  \in \calD(A) \setminus \set{\kb{0}}} \frac{\nrm{\rho - \kb{0}}_2}{\nrm{(\id_A -P)(\rho - \kb{0})}_2} < \infty.
\end{align}
where $P$ is the projection onto $\ker \calE$.

 \paragraph{Step 3} It will be more convenient in the sequel to consider  linear operators acting on $A$ as points in $\R^{2d^2}$. We use the function $f$ defined in Definition~\ref{def:f}, for which we list here some useful  properties.
 \begin{proposition}
\label{prop:f-prob}
The function $f$ defined in Definition~\ref{def:f} satisfies the following properties.
\begin{enumerate}
    \item $f$ is bijective
    \item $f(aX +bY) = af(X) + bf(Y)$ for all $X, Y\in \calL(A)$ and for all $a, b\in \R$
    \item $\nrm{f(X)}_2 = \nrm{X}_2$ for all $X\in \calX$
    \item If $Q$ is a projection onto a linear subspace $E\subset \calL(A)$, then $f(E)$ is also a linear subspace of $\R^{2d^2}$ and $f(Q(X)) = Q'(f(X))$ where $Q'$ denotes the projection onto $f(E)$.
    \item If $X$ is a compact convex subset of $\calL(A)$, then $f(X)$ is a compact convex subset of $\R^{2d^2}$ and $\partial f(X) = f(\partial X)$ where $\partial f(X)$ and $\partial X$ denote the boundaries of $f(X)$ and $X$, respectively.
\end{enumerate}
\end{proposition}
\begin{proof}
We only prove item 4 and the other items are straightforward consequence of the definition of $f$. We have$f(Q(X)) = f(\argmin_{Y\in E} \nrm{Y-X}_2)  = f(\argmin_{Y\in E} \nrm{f(Y)-f(X)}_2) = \argmin_{Y' \in f(E)} \nrm{Y' - f(X)}_2 = Q'(f(X))$.
\end{proof}
 
 Proposition~\ref{prop:f-prob} implies that 
\begin{align}
\sup_{\rho  \in \calD(A) \setminus \set{\kb{0}}} \frac{\nrm{\rho - \kb{0}}_2}{\nrm{(\id_A -P)(\rho - \kb{0})}_2} < \infty
\end{align}
if and only if
\begin{align}
\sup_{x  \in f(\calD(A) -\kb{0}) \setminus \set{0}} \frac{\nrm{x}_2}{\nrm{(\one - P')(x)}_2} < \infty, \label{eq:frac-norm}
\end{align}
where $P'$ is the projection onto $f(\ker \calE)$.

\paragraph{Step 4} 
We now provide a geometric characterization for Eq.~\eqref{eq:frac-norm} (See Fig~\ref{fig:tangent}).
\begin{figure}[h]
  \centering
  \includegraphics{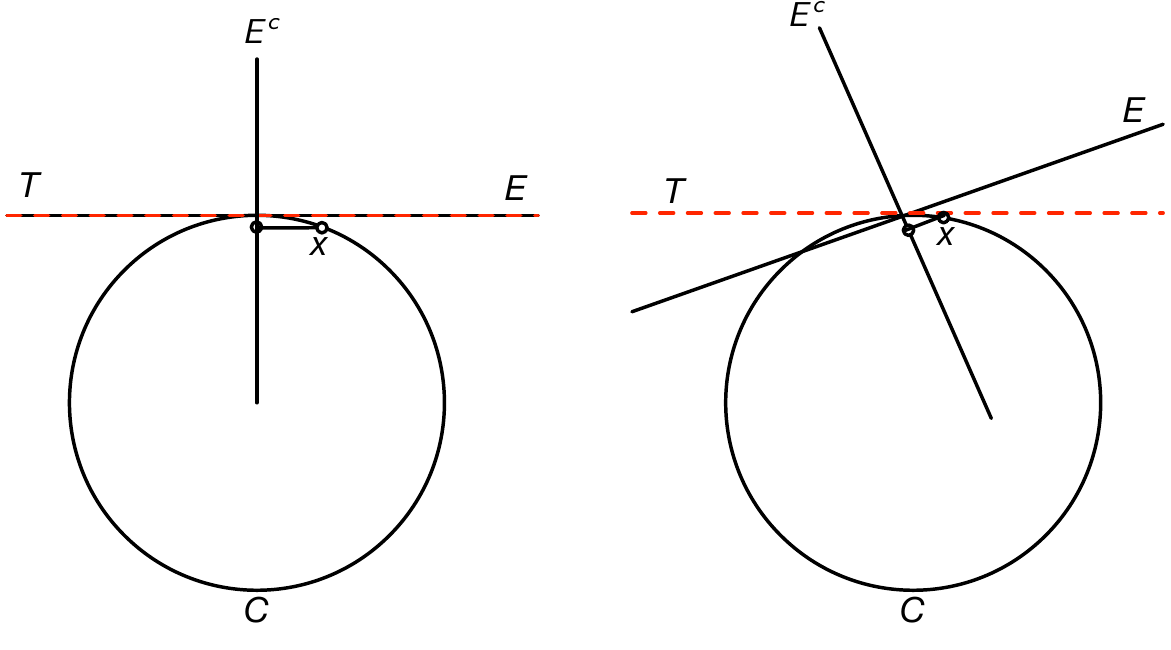}
  \caption{Illustration of Proposition~\ref{prop:bound-tang}: On the left, $\norm[2]{x}$ \emph{cannot} be uniformly bounded by $\norm[2]{(\one - P)x}$ when $x$ is close to the origin, while on the right, $\norm[2]{x}$ \emph{can} be uniformly bounded by $\norm[2]{(\one - P)x}$ when $x$ is close to the origin}
  \label{fig:tangent}
\end{figure}
\begin{proposition}
\label{prop:bound-tang}
Let $C$ be a compact  convex subset of $\R^k$ containing the origin on its boundary and $E$ be a linear subspace of $\R^k$ such that $C \cap E = \set{0}$. We assume that the boundary of $C$, $\partial C $, is a smooth manifold embedded in $\R^k$ and the tangent space of $\partial C$ at the origin is $T$. Then, upon denoting the projection onto $E$ by $P$, \begin{align}
    \sup_{x\in C\setminus\set{0}} \frac{\nrm{x}_2}{\nrm{(\one - P)x}_2} < \infty
\end{align}
if and only if
$T\cap E = \set{0}$.
\end{proposition}
\begin{proof}
See Appendix~\ref{sec:bound-tang}.
\end{proof}

\paragraph{Step 5} We show here that the tangent space at the origin of the boundary of $f(\calD(A) -\kb{0})$ is $\textnormal{span}(a_1, \cdots, a_{2d - 2})$, where $\set{a_i}_{i\in \intseq{1}{2d-2}}$ is defined in Definition~\ref{def:f}. First, note that the boundary of $f(\calD(A) \setminus \kb{0})$ is $f(\set{\kb{\phi}: \norm{\phi} = 1} \setminus \kb{0})$ because the boundary of $\calD\pr{A}$ are pure states and because of item 5 of Proposition~\ref{prop:f-prob}.  We define two maps
\begin{align}
    g: A&\to \calL(A)\\
    \ket{\phi}&\mapsto \kb{\phi}
\end{align}
and
\begin{align}
    h:\R^{2d-1}&\to A\\
    (x_1, x_2, y_2, x_3, y_3, \cdots, x_d, y_d) &\mapsto x_1 \ket{e_1} + \sum_{j=2}^d (x_j + i y_j)\ket{e_j}. 
\end{align}
Note that the coefficient of $\ket{e_1}$ is always real for all vectors in the range of $h$ as we have  freedom to choose the phase of a quantum state. Following our definition of $f$, $g$, and $h$, the $2((j-1)\times d + k) -1$ and $2((j-1)\times d + k)$ components of $(f\circ g\circ h)(x_1, x_2, y_2, x_3, y_3, \cdots, x_d, y_d)$ are $x_jx_k-y_jy_l$ and $y_jx_k + x_jy_k$, respectively. Thus, $f\circ g\circ h$ is a smooth function. We  also calculate the derivative of  $f\circ g \circ h$ at $(1, 0, \cdots, 0)$, which is represented by the matrix $[a_0|a_1| \cdots|a_{2d-2}]$, where $a_0$ is a vector and $a_1, \cdots, a_{2d-2}$ are as in Definition~\ref{def:f}.

Let $\calS \eqdef \set{x\in \R^{2d-1}: \norm[2]{x} = 1}$ be the unit sphere in $\R^{2d-1}$. The restriction of $f\circ g\circ h$ is also a smooth function.  The tangent space of $\calS$ at $(1, 0, \cdots 0)$ is the span of $(b_2, \cdots, b_{2d-1})$, where $b_1, \cdots, b_{2d-1}$ form the standard basis for $\R^{2d-1}$. Therefore, the image of the derivative of $f\circ g \circ h$ restricted to the tangent space of $\calS$ at $(1, 0, \cdots, 0)$ is the space of $(a_1, \cdots, a_{2d-2})$. Since $a_1, \cdots, a_{2d-2}$ are linearly independent and the dimension of $f(\set{\kb{\phi}: \norm{\phi} = 1} \setminus \kb{0})$ is $2d-2$, the whole tangent space of $f(\set{\kb{\phi}: \norm{\phi} = 1} \setminus \kb{0})$ at the origin should be the span of $a_1, \cdots, a_{2d-2}$.

\subsection{Proof of Proposition~\ref{prop:bound-operator}}
\label{sec:bound-operator}
Let $V/\ker V$ be the quotient space and $\pi: V \to V/\ker A$ be the quotient map. We can define a norm on $V/\ker A$ by $\nrm{\pi(v)} = \inf_{x \in \ker A} \nrm{v - x} = \nrm{v - Pv} = \nrm{(\one_V - P) v}$. By the first isomorphism theorem of linear algebra, there exists a linear isomorphism $\widetilde{A}: V/\ker A \to W$ such that $Av = (\widetilde{A}\circ \pi) v$. Since any linear operator from a finite dimensional space is bounded, we have
\begin{align}
    \nrm{Av} = \nrm{\widetilde{A}(\pi (v))} \leq \nrm{\widetilde{A}} \nrm{\pi(v)} =  \nrm{\widetilde{A}} \nrm{(\one_V - P) v}
\end{align}
and 
\begin{align}
   \nrm{(\one_V - P) v} = \nrm{\pi(v)} = \nrm{\widetilde{A}^{-1}(\widetilde{A}(\pi(v)))} \leq \nrm{\widetilde{A}^{-1}} \nrm{\widetilde{A}(\pi(v))} = \nrm{\widetilde{A}^{-1}} \nrm{Av}.
\end{align}
The result therefore holds  for $B_1 = 1/\nrm{\widetilde{A}}$ and $B_2 = \nrm{\widetilde{A}^{-1}}$.

\subsection{Proof of Proposition~\ref{prop:bound-tang}}
\label{sec:bound-tang}
\paragraph{Step 1} We first show that 
\begin{align}
        \sup_{x\in C\setminus\set{0}} \frac{\nrm{x}_2}{\nrm{(\one - P)x}_2} < \infty
\end{align}
if and only if
\begin{align}
    \sup_{x\in \partial C\setminus\set{0}} \frac{\nrm{x}_2}{\nrm{(\one - P)x}_2} < \infty,
\end{align}
If $Px = 0$ for $x\in C\setminus\set{0}$, we have $\frac{\nrm{x}_2}{\nrm{(\one - P)x}_2 } = 1$, which is bounded.
Let $x\in C\setminus \set{0}$ such that $Px \neq 0$ and define $\phi: [0, 1] \to \R^k$ by $\phi(t) \eqdef (1-t) x + t Px$. We know that $\phi^{-1}(C)$ is closed and connected, because $C$ is closed and connected (as a convex set) and $\phi$ is continuous. The only closed and connected subsets of $[0, 1]$ are of closed intervals and as $\phi(0) = x \in C$, we have $\phi^{-1}(C) = [0, a]$ for some $a \in [0, 1]$. Since $Px \neq 0$ and $E \cap C = \set{0}$, we have $Px \notin C$ and therefore $a < 1$. $\phi(a)$ is on the boundary of $C$ because $a$ is on the boundary of $\phi^{-1}(C) = [0, a]$ and $\phi$ is continuous. Now note that
\begin{align}
    P \phi(a) = P((1-a) x + a Px) = (1-a) Px + aP^2x = Px.
\end{align}
Hence,
\begin{align}
    \nrm{(1-P)\phi(a)}_2= \nrm{(1-a)x + a Px - Px }_2 = (1-a) \nrm{(1-P)x}_2 \leq \nrm{(1-P)x}_2
\end{align}
Therefore,
\begin{align}
    \frac{\nrm{x}_2}{\nrm{(\one - P)x}_2} 
    &= \frac{\sqrt{\nrm{(1-P)x}_2^2 + \nrm{Px}_2^2}}{\nrm{(\one - P)x}_2} \\
    &\stackrel{(a)}{\leq} \frac{\sqrt{\nrm{(1-P)\phi(a)}_2^2 + \nrm{P\phi(a)}_2^2}}{\nrm{(\one - P)\phi(a)}_2}\\
    &= \frac{\nrm{\phi(a)}_2}{\nrm{(\one - P)\phi(a)}_2}.
\end{align}
This completes the proof of the first step.
\paragraph{Step 2} We now show that 
\begin{align}
    \label{eq:bounded-boundry}
    \sup_{x\in \partial C\setminus\set{0}} \frac{\nrm{x}_2}{\nrm{(\one - P)x}_2} < \infty,
\end{align}
if and only if $E\cap T =\set{0}$. 

First suppose that $v\in E\cap T$ is non-zero. We will find $x\in \partial C \setminus \set{0}$ such that $\nrm{x}/\nrm{(1-P)x} \geq K$ for a given $K>0$. By definition of tangent space, there exists a smooth curve $\gamma: (-1, 1) \to \partial C$ such that $\gamma(0) = 0$ and $\gamma'(0) = v$, i.e., $\lim_{t\to 0} \gamma(t) / t = v$. There exists some $t_0 > 0$ such that $\nrm{\gamma(t)/t - v} \leq \nrm{v}/2$ for all $0 < |t| < t_0$. We therefore have $\nrm{\gamma(t)} \geq t \nrm{v}/2$. Additionally, $\lim_{t\to 0} (\one-P)\gamma t = 0$ because $\one - P$ is continuous and $(\one - P) v =  v - Pv = 0$. Thus, there exists $t_1 > 0$ such that $\nrm{(\one - P)\gamma(t)} \leq 2/(K\nrm{v})$. For any $t$ such that $0 < t < \min(t_0, t_1)$, we have
\begin{align}
    \frac{\nrm{\gamma(t)}}{\nrm{(1-P) \gamma(t)}} \geq \frac{t\nrm{v}/2}{2/(K\nrm{v}) } = K,
\end{align}
as claimed.

We now prove the other direction. Let $\calB(\epsilon) \eqdef \set{x \in \R^k: \nrm{x} < \epsilon}$ denote the open ball of radius $\epsilon$ at the origin.  To show  \eqref{eq:bounded-boundry}, it is enough to check for arbitrary small $\epsilon > 0$
\begin{align}
    \sup_{x\in (\partial C \cap \calB(\epsilon)) \setminus\set{0}} \frac{\nrm{x}_2}{\nrm{(\one - P)x}_2} < \infty
\end{align}
because $\partial C \setminus \calB(\epsilon)$ is a compact set, on which the distance from $E$ is non-zero and varies continuously. Let $Q$ denote the projection onto $T$ the tangent space of $\partial C$ at the origin. We know that for a point $x$ on $\partial C$ close to origin we have $\nrm{x} = \nrm{Q x} + o(\nrm{x})$. We can hence find an $\epsilon >0$ such that for $x \in\partial C \cap \calB(\epsilon)) \setminus\set{0} $ we have $\nrm{x} \leq 2\nrm{Qx}$. Furthermore, considering the linear map $A: T \to E^\perp$ defined by $x\mapsto (1-P)x $, it is injective. Therefore, for some constant $B > 0$, we have $\nrm{Ax} = \nrm{(1-P) x } \geq B \nrm{x}$ for all $x\in T$. We can also choose $\epsilon > 0$ such that $\nrm{(1-Q)x} \leq B/2\nrm{Qx}$.  Therefore, we have for all $x\in\partial (C \cap \calB(\epsilon)) \setminus\set{0} $
\begin{align}
    \frac{\nrm{x}}{\nrm{(\one-P) x}} 
    &\leq \frac{2\nrm{Qx}}{\nrm{(\one - P)x}}\\
    &\leq   \frac{2\nrm{Qx}}{\nrm{(\one - P)Qx} - \nrm{(\one - P)(\one-Q)x}}\\
    &\leq \frac{2\nrm{Qx}}{\nrm{(\one - P)Qx} - \nrm{(\one - P)(\one-Q)x}}\\
    &\leq \frac{2\nrm{Qx}}{C\nrm{Qx} - C\nrm{Qx}/2}\\
    &\leq 4/C.
\end{align}
\bibliographystyle{IEEEtran}
\bibliography{covert_est}
\end{document}

%% file: mehrdad_commands.tex
\newcommand{\pr}[1]{\left({#1}\right)}

\newcommand{\id}{{\mathrm{id}}}
\newcommand{\bigO}[2][]{\calO_{#1}\pr{#2}}
\newcommand{\one}{\mathbf{1}}
\newcommand{\R}{\mathbb{R}}

\newcommand{\ketbra}[2]{{\ket{#1}\bra{#2}}}
\newcommand{\kb}[1]{{\ket{#1}\bra{#1}}}
\newcommand{\ptr}[2]{\mathrm{tr}_{#1}\pr{#2}}

\newcommand{\nrm}[1]{{\norm{#1}}}
\acrodef{POVM}{Positive Operator Valued Measurement}

\acrodef{ROC}[ROC]{Receiver Operation Characteristic}

\acrodef{OOK}[OOK]{On-Off Keying}

%% file: quantum-system-model.tex
{\pgfkeys{/pgf/fpu/.try=false}%
\ifx\XFigwidth\undefined\dimen1=0pt\else\dimen1\XFigwidth\fi
\divide\dimen1 by 7994
\ifx\XFigheight\undefined\dimen3=0pt\else\dimen3\XFigheight\fi
\divide\dimen3 by 2744
\ifdim\dimen1=0pt\ifdim\dimen3=0pt\dimen1=3946sp\dimen3\dimen1
  \else\dimen1\dimen3\fi\else\ifdim\dimen3=0pt\dimen3\dimen1\fi\fi
\tikzpicture[x=+\dimen1, y=+\dimen3]
{\ifx\XFigu\undefined\catcode`\@11
\def\temp{\alloc@1\dimen\dimendef\insc@unt}\temp\XFigu\catcode`\@12\fi}
\XFigu3946sp
% Uncomment to scale line thicknesses with the same
% factor as width of the drawing.
%\pgfextractx\XFigu{\pgfqpointxy{1}{1}}
\ifdim\XFigu<0pt\XFigu-\XFigu\fi
\pgfdeclarearrow{
  name = xfiga1,
  parameters = {
    \the\pgfarrowlinewidth \the\pgfarrowlength \the\pgfarrowwidth\ifpgfarrowopen o\fi},
  defaults = {
	  line width=+7.5\XFigu, length=+120\XFigu, width=+60\XFigu},
  setup code = {
    % miter protrusion = thk * sqrt(wd^2 + (tipmv*len)^2) / (2 * wd)
    \dimen7 2.1\pgfarrowlength\pgfmathveclen{\the\dimen7}{\the\pgfarrowwidth}
    \dimen7 2\pgfarrowwidth\pgfmathdivide{\pgfmathresult}{\the\dimen7}
    \dimen7 \pgfmathresult\pgfarrowlinewidth
    \pgfarrowssettipend{+\dimen7}
    \pgfarrowssetbackend{+-\pgfarrowlength}
    \dimen9 -\pgfarrowlength\advance\dimen9 by-0.45\pgfarrowlinewidth
    \pgfarrowssetlineend{+\dimen9}
    \dimen9 -\pgfarrowlength\advance\dimen9 by-0.5\pgfarrowlinewidth
    \pgfarrowssetvisualbackend{+\dimen9}
    \pgfarrowshullpoint{+\dimen7}{+0pt}
    \pgfarrowsupperhullpoint{+-\pgfarrowlength}{+0.5\pgfarrowwidth}
    \pgfarrowssavethe\pgfarrowlinewidth
    \pgfarrowssavethe\pgfarrowlength
    \pgfarrowssavethe\pgfarrowwidth
  },
  drawing code = {\pgfsetdash{}{+0pt}
    \ifdim\pgfarrowlinewidth=\pgflinewidth\else\pgfsetlinewidth{+\pgfarrowlinewidth}\fi
    \pgfpathmoveto{\pgfqpoint{-\pgfarrowlength}{-0.5\pgfarrowwidth}}
    \pgfpathlineto{\pgfqpoint{0pt}{0pt}}
    \pgfpathlineto{\pgfqpoint{-\pgfarrowlength}{0.5\pgfarrowwidth}}
    \pgfpathclose
    \ifpgfarrowopen\pgfusepathqstroke\else\pgfsetfillcolor{.}
	\ifdim\pgfarrowlinewidth>0pt\pgfusepathqfillstroke\else\pgfusepathqfill\fi\fi
  }
}
\definecolor{green3}{rgb}{0,0.82,0}
\definecolor{pink1}{rgb}{1,0.5,0.5}
\definecolor{pink4}{rgb}{1,0.88,0.88}
\definecolor{gold}{rgb}{1,0.84,0}
\clip(878,-4822) rectangle (8872,-2078);
\tikzset{inner sep=+0pt, outer sep=+0pt}
\pgfsetlinewidth{+15\XFigu}
\pgfsetdash{{+90\XFigu}{+90\XFigu}}{++0pt}
\pgfsetstrokecolor{black}
\pgfsetfillcolor{black!5}
\filldraw (8850,-4800) [rounded corners=+105\XFigu] rectangle (6600,-3600);
\filldraw (900,-3450)--(900,-2100)--(8850,-2100)--(8850,-3300)--(6600,-3300)--(6600,-2400)--(3300,-2400)
  --(3300,-3450)--(1200,-3450);
\pgfsetbeveljoin
\pgfsetfillcolor{green3!50}
\fill (8349,-2816)--(8349,-2817)--(8349,-2821)--(8350,-2831)--(8352,-2847)--(8354,-2870)
  --(8357,-2896)--(8359,-2921)--(8361,-2944)--(8363,-2960)--(8364,-2970)--(8364,-2974)
  --(8364,-2975)--(8365,-2975)--(8369,-2976)--(8380,-2979)--(8398,-2983)--(8423,-2989)
  --(8451,-2996)--(8479,-3003)--(8504,-3009)--(8522,-3013)--(8533,-3016)--(8537,-3017)
  --(8538,-3017)--(8539,-3017)--(8544,-3016)--(8555,-3014)--(8574,-3011)--(8600,-3007)
  --(8631,-3001)--(8661,-2996)--(8687,-2992)--(8706,-2989)--(8717,-2987)--(8722,-2986)
  --(8723,-2986)--(8723,-2985)--(8724,-2981)--(8725,-2971)--(8726,-2955)--(8729,-2932)
  --(8732,-2906)--(8734,-2880)--(8737,-2857)--(8738,-2841)--(8739,-2831)--(8740,-2827)
  --(8740,-2826)--(8740,-2825)--(8738,-2821)--(8734,-2811)--(8727,-2795)--(8717,-2772)
  --(8706,-2746)--(8696,-2721)--(8686,-2698)--(8679,-2682)--(8675,-2672)--(8673,-2668)
  --(8673,-2667)--(8670,-2666)--(8662,-2664)--(8650,-2661)--(8631,-2657)--(8606,-2651)
  --(8578,-2644)--(8547,-2637)--(8515,-2629)--(8487,-2622)--(8462,-2616)--(8443,-2612)
  --(8431,-2609)--(8423,-2607)--(8420,-2606)--(8420,-2607)--(8419,-2610)--(8416,-2619)
  --(8411,-2634)--(8403,-2655)--(8394,-2682)--(8384,-2711)--(8375,-2740)--(8366,-2767)
  --(8358,-2788)--(8353,-2803)--(8350,-2812)--(8349,-2815)--cycle;
\fill (1067,-2861)--(1067,-2862)--(1067,-2868)--(1066,-2881)--(1065,-2901)--(1064,-2927)
  --(1063,-2952)--(1062,-2972)--(1061,-2985)--(1061,-2991)--(1061,-2992)--(1062,-2992)
  --(1065,-2994)--(1075,-2998)--(1090,-3004)--(1113,-3013)--(1141,-3025)--(1173,-3038)
  --(1204,-3051)--(1232,-3063)--(1255,-3072)--(1270,-3078)--(1280,-3082)--(1283,-3084)
  --(1284,-3084)--(1285,-3084)--(1290,-3082)--(1301,-3079)--(1321,-3074)--(1347,-3067)
  --(1378,-3058)--(1408,-3050)--(1434,-3043)--(1454,-3038)--(1465,-3035)--(1470,-3033)
  --(1471,-3033)--(1471,-3032)--(1471,-3027)--(1472,-3016)--(1473,-2996)--(1474,-2970)
  --(1476,-2939)--(1478,-2908)--(1479,-2882)--(1480,-2862)--(1481,-2851)--(1481,-2846)
  --(1481,-2845)--(1481,-2844)--(1479,-2840)--(1475,-2830)--(1468,-2814)--(1458,-2791)
  --(1447,-2765)--(1437,-2739)--(1427,-2716)--(1420,-2700)--(1416,-2690)--(1414,-2686)
  --(1414,-2685)--(1413,-2685)--(1409,-2684)--(1399,-2682)--(1382,-2678)--(1357,-2673)
  --(1327,-2667)--(1292,-2660)--(1258,-2653)--(1228,-2647)--(1203,-2642)--(1186,-2638)
  --(1176,-2636)--(1172,-2635)--(1171,-2635)--(1171,-2636)--(1169,-2639)--(1165,-2649)
  --(1157,-2665)--(1147,-2688)--(1134,-2716)--(1119,-2748)--(1104,-2780)--(1091,-2808)
  --(1081,-2831)--(1073,-2847)--(1069,-2857)--(1067,-2860)--cycle;
\pgfsetmiterjoin
\pgfsetdash{}{+0pt}
\pgfsetfillcolor{white}
\filldraw (1963,-2610) rectangle (3062,-3060);
\pgfsetlinewidth{+7.5\XFigu}
\filldraw  (8616,-2364) ellipse [x radius=+27,y radius=+45,rotate=+354];
\pgfsetlinewidth{+15\XFigu}
\filldraw (6900,-2610) rectangle (8100,-3060);
\pgfsetlinewidth{+7.5\XFigu}
\pgfsetfillcolor{pink4!50}
\filldraw  (1488,-2843) circle [radius=+40];
\pgfsetlinewidth{+15\XFigu}
\pgfsetfillcolor{black!15}
\filldraw (4350,-4050) rectangle (5550,-4650);
\filldraw (4350,-2550) rectangle (5550,-3150);
\pgfsetfillcolor{pink4!50}
\fill  (8503,-3893) circle [radius=+229];
\pgfsetfillcolor{gold!65!black}
\fill  (8241,-3928) ellipse [x radius=+19,y radius=+56,rotate=+355];
\pgfsetfillcolor{white}
\filldraw (6900,-4110) rectangle (8100,-4560);
\pgfsetfillcolor{pink4!50}
\fill  (8530,-2415) circle [radius=+216];
\pgfsetlinewidth{+7.5\XFigu}
\pgfsetfillcolor{white}
\filldraw  (8440,-3854) ellipse [x radius=+37,y radius=+70,rotate=+356];
\pgfsetbeveljoin
\pgfsetfillcolor{green3!50}
\fill (8379,-3174)--(8380,-3174)--(8387,-3175)--(8403,-3177)--(8424,-3180)--(8445,-3183)
  --(8461,-3185)--(8468,-3186)--(8469,-3186)--(8470,-3185)--(8472,-3177)--(8477,-3164)
  --(8481,-3150)--(8483,-3139)--(8480,-3131)--(8475,-3127)--(8466,-3123)--(8453,-3120)
  --(8439,-3117)--(8424,-3114)--(8414,-3112)--(8409,-3111)--(8408,-3111)--(8408,-3112)
  --(8405,-3117)--(8400,-3128)--(8394,-3143)--(8387,-3157)--(8382,-3168)--(8379,-3173)--cycle;
\fill (1180,-3287)--(1181,-3287)--(1187,-3289)--(1199,-3291)--(1216,-3296)--(1232,-3300)
  --(1244,-3302)--(1250,-3304)--(1251,-3304)--(1251,-3303)--(1252,-3297)--(1254,-3285)
  --(1257,-3268)--(1259,-3250)--(1261,-3238)--(1262,-3232)--(1262,-3231)--(1261,-3231)
  --(1256,-3230)--(1244,-3228)--(1227,-3226)--(1210,-3223)--(1198,-3221)--(1193,-3220)
  --(1192,-3220)--(1192,-3221)--(1191,-3226)--(1189,-3238)--(1186,-3254)--(1183,-3269)
  --(1181,-3281)--(1180,-3286)--cycle;
\pgfsetfillcolor{pink1!50}
\fill (8578,-4690)--(8579,-4690)--(8584,-4690)--(8596,-4690)--(8612,-4690)--(8628,-4690)
  --(8640,-4690)--(8645,-4690)--(8646,-4690)--(8647,-4689)--(8651,-4683)--(8658,-4671)
  --(8665,-4660)--(8669,-4654)--(8670,-4653)--(8669,-4653)--(8662,-4651)--(8647,-4647)
  --(8627,-4643)--(8607,-4638)--(8592,-4634)--(8585,-4632)--(8584,-4632)--(8584,-4634)
  --(8583,-4643)--(8581,-4661)--(8579,-4679)--(8578,-4688)--cycle;
\fill (8346,-4698)--(8347,-4698)--(8353,-4697)--(8365,-4694)--(8382,-4691)--(8398,-4687)
  --(8410,-4684)--(8416,-4683)--(8417,-4683)--(8417,-4682)--(8418,-4677)--(8419,-4666)
  --(8421,-4651)--(8422,-4636)--(8423,-4625)--(8424,-4620)--(8424,-4619)--(8423,-4619)
  --(8416,-4622)--(8402,-4627)--(8383,-4634)--(8364,-4641)--(8350,-4646)--(8343,-4649)
  --(8342,-4649)--(8342,-4650)--(8343,-4659)--(8344,-4674)--(8345,-4688)--(8346,-4697)--cycle;
\pgfsetfillcolor{green3!50}
\fill (1345,-3217)--(1346,-3217)--(1352,-3217)--(1364,-3216)--(1382,-3215)--(1399,-3213)
  --(1411,-3212)--(1417,-3212)--(1418,-3212)--(1418,-3211)--(1416,-3203)--(1414,-3189)
  --(1412,-3176)--(1410,-3168)--(1410,-3167)--(1408,-3167)--(1398,-3168)--(1380,-3169)
  --(1363,-3169)--(1353,-3170)--(1351,-3170)--(1351,-3171)--(1350,-3179)--(1348,-3194)
  --(1346,-3208)--(1345,-3216)--cycle;
\pgfsetfillcolor{white}
\filldraw  (8461,-2345) ellipse [x radius=+37,y radius=+70,rotate=+356];
\pgfsetlinewidth{+15\XFigu}
\draw (8407,-2450)--(8362,-2694);
\pgfsetlinewidth{+7.5\XFigu}
\draw  (8534,-2436) circle [radius=+216];
\pgfsetfillcolor{green3!50}
\fill (8596,-3166)--(8597,-3166)--(8604,-3168)--(8618,-3173)--(8637,-3179)--(8656,-3184)
  --(8670,-3189)--(8677,-3191)--(8678,-3191)--(8678,-3190)--(8678,-3186)--(8678,-3177)
  --(8678,-3166)--(8678,-3153)--(8677,-3142)--(8675,-3133)--(8672,-3126)--(8667,-3120)
  --(8660,-3115)--(8650,-3109)--(8638,-3104)--(8628,-3098)--(8619,-3095)--(8616,-3093)
  --(8615,-3093)--(8615,-3094)--(8613,-3100)--(8610,-3112)--(8606,-3130)--(8601,-3147)
  --(8598,-3159)--(8596,-3165)--cycle;
\pgfsetfillcolor{white}
\filldraw  (8578,-3850) ellipse [x radius=+37,y radius=+70,rotate=+356];
\pgfsetfillcolor{pink1!50}
\fill (8312,-4317)--(8312,-4318)--(8313,-4322)--(8314,-4333)--(8315,-4350)--(8318,-4374)
  --(8321,-4401)--(8323,-4428)--(8326,-4452)--(8327,-4469)--(8328,-4480)--(8329,-4484)
  --(8329,-4485)--(8330,-4485)--(8335,-4486)--(8346,-4489)--(8365,-4493)--(8391,-4500)
  --(8421,-4507)--(8450,-4513)--(8476,-4520)--(8495,-4524)--(8506,-4527)--(8511,-4528)
  --(8512,-4528)--(8513,-4528)--(8516,-4527)--(8524,-4526)--(8538,-4524)--(8558,-4521)
  --(8583,-4517)--(8611,-4512)--(8638,-4508)--(8663,-4504)--(8683,-4501)--(8697,-4499)
  --(8705,-4498)--(8708,-4497)--(8709,-4497)--(8709,-4496)--(8710,-4492)--(8711,-4481)
  --(8713,-4464)--(8715,-4440)--(8718,-4412)--(8721,-4385)--(8723,-4361)--(8725,-4344)
  --(8726,-4333)--(8727,-4329)--(8727,-4328)--(8727,-4327)--(8725,-4323)--(8721,-4312)
  --(8713,-4295)--(8703,-4271)--(8692,-4243)--(8681,-4216)--(8671,-4192)--(8663,-4175)
  --(8659,-4164)--(8657,-4160)--(8657,-4159)--(8653,-4158)--(8646,-4156)--(8632,-4153)
  --(8612,-4148)--(8586,-4142)--(8556,-4135)--(8523,-4127)--(8489,-4119)--(8459,-4112)
  --(8433,-4106)--(8413,-4101)--(8399,-4098)--(8392,-4096)--(8388,-4095)--(8388,-4096)
  --(8387,-4099)--(8383,-4108)--(8378,-4124)--(8370,-4147)--(8361,-4175)--(8350,-4206)
  --(8339,-4237)--(8330,-4265)--(8322,-4288)--(8317,-4304)--(8313,-4313)--(8312,-4316)--cycle;
\pgfsetmiterjoin
\pgfsetfillcolor{gold!80}
\fill (8242,-3873)--(8170,-3899)--(8150,-3932)--(8166,-3962)--(8229,-3987)--(8216,-3929)
  --(8222,-3907)--(8233,-3894)--cycle;
\pgfsetfillcolor{black}
\pgftext[base,left,at=\pgfqpointxy{5625}{-4200}] {\fontsize{14}{16.8}\normalfont $\calE_\theta^{\pn}(\phi_{A^n})$}
\pgftext[base,at=\pgfqpointxy{2513}{-2910}] {\fontsize{13}{15.6}\normalfont Alice}
\pgftext[base,at=\pgfqpointxy{7500}{-4410}] {\fontsize{14}{16.8}\normalfont Willie}
\filldraw  (8459,-2364) circle [radius=+16];
\pgfsetfillcolor{pink4!50}
\fill  (8366,-2608) circle [radius=+39];
\pgfsetcolor{black}
\filldraw  (8610,-2380) circle [radius=+9];
\draw (8449,-2975)--(8450,-3168);
\draw (8599,-2996)--(8606,-3168);
\pgfsetbeveljoin
\draw (8635,-2625)--(8635,-2626)--(8637,-2629)--(8641,-2637)--(8646,-2651)--(8654,-2668)
  --(8663,-2688)--(8671,-2709)--(8679,-2728)--(8686,-2746)--(8692,-2762)--(8697,-2777)
  --(8701,-2790)--(8703,-2802)--(8706,-2813)--(8707,-2827)--(8708,-2840)--(8708,-2854)
  --(8707,-2870)--(8705,-2887)--(8702,-2906)--(8699,-2926)--(8696,-2945)--(8693,-2960)
  --(8691,-2970)--(8690,-2974)--(8690,-2975);
\draw (8375,-2941)--(8376,-2941)--(8380,-2943)--(8391,-2947)--(8407,-2953)--(8427,-2960)
  --(8448,-2967)--(8468,-2974)--(8486,-2979)--(8502,-2984)--(8516,-2987)--(8529,-2989)
  --(8541,-2990)--(8553,-2991)--(8566,-2991)--(8579,-2990)--(8593,-2988)--(8609,-2985)
  --(8627,-2981)--(8646,-2977)--(8664,-2972)--(8678,-2969)--(8687,-2966)--(8691,-2965)
  --(8692,-2965);
\draw (8686,-2739)--(8688,-2739)--(8696,-2738)--(8713,-2738)--(8732,-2737)--(8749,-2737)
  --(8762,-2737)--(8771,-2739)--(8777,-2742)--(8781,-2747)--(8782,-2756)--(8779,-2769)
  --(8774,-2784)--(8770,-2794)--(8769,-2796);
\draw (8524,-2456)--(8524,-2457)--(8526,-2464)--(8529,-2474)--(8533,-2482)--(8536,-2488)
  --(8541,-2492)--(8544,-2494)--(8549,-2496)--(8554,-2497)--(8560,-2498)--(8566,-2498)
  --(8572,-2497)--(8577,-2496)--(8581,-2494)--(8586,-2491)--(8591,-2486)--(8595,-2480)
  --(8597,-2474)--(8599,-2468)--(8599,-2463)--(8598,-2456)--(8595,-2450)--(8591,-2444)
  --(8587,-2439)--(8582,-2436)--(8577,-2434)--(8571,-2432)--(8563,-2431)--(8556,-2431)
  --(8555,-2431);
\draw (8383,-2760)--(8382,-2760)--(8377,-2759)--(8366,-2756)--(8351,-2752)--(8334,-2748)
  --(8319,-2744)--(8307,-2740)--(8297,-2737)--(8289,-2734)--(8283,-2730)--(8278,-2727)
  --(8273,-2723)--(8269,-2719)--(8266,-2714)--(8264,-2710)--(8263,-2705)--(8262,-2701)
  --(8261,-2697)--(8261,-2693)--(8262,-2690)--(8263,-2684)--(8265,-2679)--(8268,-2674)
  --(8272,-2669)--(8276,-2664)--(8282,-2660)--(8287,-2656)--(8293,-2651)--(8301,-2645)
  --(8311,-2638)--(8322,-2630)--(8331,-2624)--(8336,-2621)--(8337,-2620);
\draw (8591,-2277)--(8592,-2278)--(8598,-2285)--(8606,-2292)--(8612,-2297)--(8618,-2299)
  --(8623,-2298)--(8630,-2297)--(8639,-2294)--(8648,-2290)--(8655,-2287)--(8656,-2286);
\draw (8420,-2629)--(8420,-2630)--(8417,-2636)--(8413,-2648)--(8406,-2665)--(8399,-2685)
  --(8393,-2703)--(8387,-2720)--(8383,-2736)--(8380,-2750)--(8378,-2764)--(8376,-2777)
  --(8375,-2790)--(8374,-2805)--(8374,-2823)--(8374,-2843)--(8374,-2865)--(8374,-2889)
  --(8374,-2912)--(8375,-2931)--(8375,-2943)--(8375,-2948)--(8375,-2949);
\draw (8447,-3169)--(8445,-3170)--(8433,-3175)--(8416,-3182)--(8402,-3186)--(8392,-3187)
  --(8386,-3185)--(8384,-3181)--(8382,-3175)--(8382,-3167)--(8384,-3159)--(8386,-3150)
  --(8389,-3141)--(8393,-3134)--(8397,-3128)--(8403,-3122)--(8412,-3118)--(8424,-3115)
  --(8437,-3113)--(8445,-3112)--(8447,-3112);
\draw (8599,-3165)--(8600,-3165)--(8607,-3166)--(8621,-3166)--(8636,-3167)--(8650,-3167)
  --(8660,-3166)--(8666,-3164)--(8670,-3162)--(8671,-3158)--(8671,-3153)--(8669,-3147)
  --(8666,-3141)--(8661,-3134)--(8656,-3128)--(8650,-3123)--(8645,-3119)--(8637,-3115)
  --(8629,-3112)--(8619,-3110)--(8608,-3109)--(8600,-3108)--(8599,-3108);
\pgftext[base,at=\pgfqpointxy{7500}{-2910}] {\fontsize{14}{16.8}\normalfont Alice}
\draw (1193,-3023)--(1193,-3280);
\draw (1338,-3052)--(1345,-3226);
\pgfsetlinewidth{+15\XFigu}
\pgfsetarrows{[line width=7.5\XFigu]}
\pgfsetarrowsend{xfiga1}
\draw (4950,-3450)--(4950,-3150);
\draw (4950,-3750)--(4950,-4050);
\draw (5550,-4350)--(6900,-4350);
\draw (5550,-2850)--(6900,-2850);
\pgfsetmiterjoin
\draw (3900,-2850)--(3900,-4350)--(4350,-4350);
\pgfsetlinewidth{+7.5\XFigu}
\pgfsetarrowsend{}
\draw (8419,-4485)--(8420,-4688);
\draw (8578,-4507)--(8583,-4688);
\draw (8233,-3985)--(8163,-3959);
\draw (8246,-3873)--(8169,-3898);
\pgfsetbeveljoin
\draw (8388,-4119)--(8388,-4120)--(8386,-4123)--(8383,-4132)--(8378,-4145)--(8372,-4161)
  --(8365,-4178)--(8360,-4194)--(8355,-4209)--(8351,-4223)--(8347,-4236)--(8345,-4249)
  --(8343,-4261)--(8341,-4273)--(8340,-4285)--(8339,-4298)--(8339,-4312)--(8338,-4329)
  --(8338,-4348)--(8338,-4368)--(8339,-4390)--(8339,-4411)--(8339,-4430)--(8340,-4444)
  --(8340,-4453)--(8340,-4456)--(8340,-4457);
\draw (8357,-4173)--(8356,-4172)--(8350,-4169)--(8338,-4162)--(8321,-4152)--(8302,-4141)
  --(8285,-4130)--(8270,-4121)--(8257,-4112)--(8247,-4104)--(8238,-4096)--(8229,-4088)
  --(8221,-4079)--(8213,-4068)--(8204,-4056)--(8194,-4042)--(8184,-4027)--(8175,-4012)
  --(8169,-4002)--(8166,-3998)--(8166,-3997);
\draw (8340,-4448)--(8341,-4448)--(8346,-4450)--(8357,-4455)--(8374,-4461)--(8396,-4469)
  --(8418,-4477)--(8439,-4484)--(8458,-4490)--(8475,-4495)--(8489,-4498)--(8503,-4500)
  --(8516,-4502)--(8529,-4502)--(8542,-4502)--(8556,-4501)--(8571,-4499)--(8588,-4496)
  --(8607,-4492)--(8627,-4487)--(8646,-4483)--(8661,-4479)--(8671,-4476)--(8675,-4475)
  --(8676,-4475);
\draw (8579,-4689)--(8580,-4689)--(8588,-4690)--(8602,-4690)--(8619,-4691)--(8633,-4691)
  --(8643,-4690)--(8650,-4688)--(8654,-4685)--(8656,-4681)--(8656,-4676)--(8654,-4670)
  --(8650,-4663)--(8645,-4657)--(8640,-4650)--(8634,-4645)--(8628,-4641)--(8620,-4636)
  --(8611,-4633)--(8600,-4631)--(8588,-4630)--(8581,-4629)--(8579,-4629);
\draw (8616,-4114)--(8616,-4115)--(8618,-4118)--(8622,-4127)--(8628,-4141)--(8636,-4160)
  --(8645,-4182)--(8654,-4203)--(8662,-4224)--(8670,-4243)--(8676,-4260)--(8680,-4276)
  --(8684,-4289)--(8687,-4302)--(8690,-4315)--(8691,-4329)--(8692,-4343)--(8692,-4358)
  --(8691,-4374)--(8689,-4393)--(8686,-4413)--(8683,-4434)--(8679,-4453)--(8676,-4469)
  --(8674,-4480)--(8673,-4484)--(8673,-4485);
\draw (8670,-4236)--(8672,-4236)--(8681,-4235)--(8698,-4234)--(8718,-4234)--(8736,-4233)
  --(8750,-4234)--(8759,-4236)--(8765,-4239)--(8770,-4244)--(8770,-4254)--(8767,-4267)
  --(8762,-4283)--(8758,-4294)--(8757,-4296);
\draw (8415,-4686)--(8414,-4686)--(8408,-4689)--(8396,-4694)--(8383,-4699)--(8371,-4703)
  --(8362,-4705)--(8355,-4705)--(8351,-4703)--(8348,-4698)--(8347,-4692)--(8347,-4684)
  --(8348,-4675)--(8350,-4666)--(8353,-4657)--(8357,-4649)--(8362,-4643)--(8368,-4636)
  --(8378,-4632)--(8390,-4629)--(8404,-4627)--(8413,-4626)--(8415,-4626);
\draw (8578,-3753)--(8577,-3754)--(8573,-3758)--(8567,-3763)--(8562,-3766)--(8559,-3769)
  --(8555,-3772)--(8550,-3775)--(8544,-3779)--(8539,-3782)--(8538,-3783);
\draw (8421,-3747)--(8422,-3748)--(8426,-3752)--(8432,-3757)--(8437,-3760)--(8440,-3763)
  --(8444,-3766)--(8449,-3769)--(8455,-3773)--(8460,-3776)--(8461,-3777);
\draw (8533,-3983)--(8533,-3982)--(8536,-3976)--(8539,-3966)--(8541,-3957)--(8541,-3951)
  --(8540,-3945)--(8538,-3940)--(8536,-3936)--(8532,-3932)--(8528,-3928)--(8523,-3924)
  --(8518,-3921)--(8513,-3919)--(8508,-3918)--(8504,-3918)--(8499,-3918)--(8494,-3919)
  --(8489,-3921)--(8485,-3923)--(8481,-3926)--(8478,-3929)--(8475,-3933)--(8472,-3939)
  --(8471,-3946)--(8471,-3954)--(8472,-3961)--(8474,-3967)--(8476,-3972)--(8480,-3976)
  --(8486,-3981)--(8490,-3984)--(8491,-3985);
\pgftext[base,at=\pgfqpointxy{4950}{-3675}] {\fontsize{14}{16.8}\normalfont $\theta$}
\pgftext[base,at=\pgfqpointxy{4950}{-4425}] {\fontsize{14}{16.8}\normalfont $\calE_\theta$}
\pgftext[base,at=\pgfqpointxy{4950}{-2925}] {\fontsize{14}{16.8}\normalfont $\text{id}_R\circ\calN_\theta$}
\pgftext[base,at=\pgfqpointxy{3900}{-2700}] {\fontsize{14}{16.8}\normalfont $\ket{\phi}_{RA^n}$}
\pgftext[base,at=\pgfqpointxy{6150}{-2700}] {\fontsize{14}{16.8}\normalfont $\ket{\psi_\theta}_{RA^n}$}
\draw (1193,-3280)--(1194,-3280)--(1200,-3283)--(1211,-3288)--(1225,-3293)--(1236,-3297)
  --(1245,-3299)--(1251,-3299)--(1255,-3297)--(1258,-3292)--(1259,-3286)--(1259,-3279)
  --(1258,-3270)--(1256,-3261)--(1253,-3253)--(1249,-3245)--(1245,-3239)--(1238,-3233)
  --(1229,-3229)--(1217,-3226)--(1204,-3224)--(1195,-3223)--(1193,-3223);
\draw (1341,-3221)--(1342,-3221)--(1350,-3222)--(1363,-3222)--(1379,-3223)--(1393,-3223)
  --(1403,-3222)--(1410,-3221)--(1414,-3218)--(1416,-3215)--(1417,-3211)--(1417,-3207)
  --(1416,-3202)--(1413,-3197)--(1410,-3192)--(1407,-3186)--(1403,-3182)--(1399,-3178)
  --(1394,-3175)--(1387,-3170)--(1377,-3167)--(1365,-3166)--(1352,-3164)--(1343,-3164)
  --(1341,-3164);
\draw (1079,-2966)--(1080,-2966)--(1085,-2969)--(1095,-2976)--(1112,-2986)--(1133,-2998)
  --(1154,-3011)--(1175,-3022)--(1194,-3031)--(1211,-3038)--(1226,-3044)--(1240,-3048)
  --(1254,-3050)--(1267,-3051)--(1282,-3051)--(1297,-3050)--(1315,-3047)--(1335,-3043)
  --(1357,-3038)--(1380,-3031)--(1402,-3025)--(1420,-3020)--(1432,-3016)--(1437,-3014)
  --(1438,-3014);
\draw (1164,-2675)--(1163,-2676)--(1160,-2680)--(1152,-2689)--(1141,-2703)--(1129,-2718)
  --(1119,-2732)--(1110,-2746)--(1103,-2759)--(1097,-2771)--(1093,-2784)--(1090,-2796)
  --(1088,-2808)--(1086,-2823)--(1084,-2840)--(1083,-2860)--(1081,-2882)--(1081,-2906)
  --(1080,-2929)--(1079,-2947)--(1079,-2960)--(1079,-2965)--(1079,-2966);
\draw (1336,-2315)--(1337,-2315)--(1342,-2315)--(1349,-2314)--(1353,-2314)--(1357,-2313)
  --(1359,-2311)--(1362,-2308)--(1364,-2304)--(1367,-2301)--(1367,-2300);
\draw (1282,-2334)--(1283,-2334)--(1288,-2334)--(1295,-2333)--(1299,-2333)--(1302,-2331)
  --(1305,-2329)--(1307,-2327)--(1310,-2322)--(1312,-2319)--(1312,-2318);
\draw (1420,-2724)--(1421,-2724)--(1427,-2725)--(1438,-2726)--(1455,-2729)--(1472,-2732)
  --(1488,-2734)--(1501,-2737)--(1510,-2741)--(1517,-2744)--(1522,-2748)--(1526,-2754)
  --(1528,-2761)--(1526,-2771)--(1523,-2783)--(1518,-2796)--(1513,-2808)--(1510,-2814)
  --(1510,-2815);
\draw (1380,-2671)--(1380,-2672)--(1383,-2677)--(1388,-2688)--(1396,-2705)--(1406,-2726)
  --(1416,-2748)--(1425,-2770)--(1433,-2789)--(1439,-2806)--(1444,-2821)--(1447,-2835)
  --(1450,-2848)--(1452,-2862)--(1453,-2876)--(1453,-2891)--(1452,-2907)--(1450,-2926)
  --(1448,-2947)--(1445,-2969)--(1442,-2990)--(1439,-3006)--(1437,-3017)--(1436,-3022)
  --(1436,-3023);
\pgfsetfillcolor{pink4!50}
\fill  (1255,-2468) ellipse [x radius=+217,y radius=+218];
\pgfsetlinewidth{+30\XFigu}
\draw (8173,-3897) arc[start angle=+132, end angle=+208, radius=+50.6];
\fill  (8131,-3942) circle [radius=+42];
\pgfsetlinewidth{+7.5\XFigu}
\draw  (8242,-3928) ellipse [x radius=+23,y radius=+56,rotate=+355];
\pgfsetcolor{black}
\filldraw  (8439,-3860) circle [radius=+16];
\filldraw  (8576,-3860) circle [radius=+16];
\pgfsetmiterjoin
\pgfsetarrowsend{xfiga1}
\pgfsetlinewidth{+15\XFigu}
\draw (4050,-2850)--(4200,-2850)--(4350,-2850);
\pgfsetarrowsend{}
\draw (3075,-2850)--(4050,-2850);
\pgfsetlinewidth{+7.5\XFigu}
\draw  (8146,-3942) circle [radius=+42];
\pgfsetlinewidth{+15\XFigu}
\pgfsetstrokecolor{gold}
\draw  (8444,-2337) circle [radius=+113];
\pgfsetfillcolor{pink4!50}
\fill  (8744,-2841) circle [radius=+39];
\pgfsetlinewidth{+7.5\XFigu}
\pgfsetstrokecolor{black}
\draw  (1493,-2843) circle [radius=+40];
\draw  (8379,-2608) circle [radius=+39];
\fill  (8731,-4343) circle [radius=+42];
\draw  (8509,-3915) circle [radius=+229];
\pgfsetfillcolor{white}
\filldraw  (1329,-2395) ellipse [x radius=+27,y radius=+55,rotate=+10];
\draw  (1279,-2480) ellipse [x radius=+217,y radius=+218];
\pgfsetfillcolor{pink4!50}
\filldraw  (1208,-2483) circle [radius=+40];
\pgfsetfillcolor{white}
\filldraw  (1397,-2371) ellipse [x radius=+25,y radius=+52,rotate=+10];
\draw  (8444,-2337) circle [radius=+127];
\pgfsetcolor{black}
\filldraw  (1397,-2362) circle [radius=+6];
\filldraw  (1328,-2387) circle [radius=+6];
\draw  (8444,-2337) circle [radius=+99];
\draw  (8742,-2836) circle [radius=+39];
\pgfsetbeveljoin
\draw (1202,-2758)--(1202,-2757)--(1200,-2753)--(1197,-2743)--(1192,-2729)--(1186,-2710)
  --(1180,-2692)--(1174,-2674)--(1170,-2658)--(1166,-2645)--(1164,-2633)--(1162,-2623)
  --(1161,-2613)--(1161,-2600)--(1163,-2587)--(1167,-2572)--(1173,-2556)--(1180,-2540)
  --(1186,-2526)--(1189,-2519)--(1190,-2518);
\draw (1410,-2440)--(1412,-2440)--(1420,-2440)--(1430,-2441)--(1438,-2443)--(1442,-2447)
  --(1445,-2451)--(1446,-2455)--(1448,-2460)--(1448,-2466)--(1448,-2471)--(1448,-2477)
  --(1447,-2481)--(1445,-2485)--(1442,-2489)--(1437,-2492)--(1432,-2494)--(1426,-2495)
  --(1421,-2495)--(1417,-2494)--(1412,-2491)--(1406,-2487)--(1400,-2481)--(1396,-2476)
  --(1395,-2475);
\filldraw  (1341,-2400) circle [radius=+16];
\filldraw  (1407,-2373) circle [radius=+16];
\draw  (8729,-4337) circle [radius=+42];
\draw  (1204,-2483) circle [radius=+40];
\filldraw  (8612,-2374) circle [radius=+11];
\pgfsetfillcolor{pink1!70}
\fill (8625,-4350)--(8626,-4365)--(8625,-4381)--(8623,-4397)--(8621,-4415)--(8617,-4432)
  --(8614,-4449)--(8611,-4462)--(8609,-4471)--(8608,-4474)--(8608,-4475)--(8610,-4475)
  --(8620,-4472)--(8638,-4468)--(8656,-4465)--(8666,-4462)--(8668,-4462)--(8668,-4461)
  --(8668,-4458)--(8669,-4450)--(8670,-4437)--(8671,-4420)--(8672,-4399)--(8673,-4378)
  --(8674,-4358)--(8675,-4339)--(8675,-4323)--(8674,-4308)--(8672,-4295)--(8670,-4282)
  --(8666,-4269)--(8661,-4255)--(8655,-4239)--(8649,-4223)--(8641,-4206)--(8634,-4189)
  --(8628,-4176)--(8623,-4165)--(8620,-4159)--(8619,-4156)--(8618,-4156)--(8612,-4157)
  --(8600,-4158)--(8589,-4159)--(8583,-4160)--(8582,-4160)--(8582,-4161)--(8583,-4165)
  --(8586,-4174)--(8590,-4189)--(8595,-4209)--(8601,-4233)--(8607,-4257)--(8612,-4281)
  --(8617,-4302)--(8621,-4320)--(8623,-4336)--cycle;
\endtikzpicture}%

%% file: quantum-covert-estimation__jsait_.bbl
\begin{thebibliography}{10}
\providecommand{\url}[1]{#1}
\csname url@samestyle\endcsname
\providecommand{\newblock}{\relax}
\providecommand{\bibinfo}[2]{#2}
\providecommand{\BIBentrySTDinterwordspacing}{\spaceskip=0pt\relax}
\providecommand{\BIBentryALTinterwordstretchfactor}{4}
\providecommand{\BIBentryALTinterwordspacing}{\spaceskip=\fontdimen2\font plus
\BIBentryALTinterwordstretchfactor\fontdimen3\font minus
  \fontdimen4\font\relax}
\providecommand{\BIBforeignlanguage}[2]{{%
\expandafter\ifx\csname l@#1\endcsname\relax
\typeout{** WARNING: IEEEtran.bst: No hyphenation pattern has been}%
\typeout{** loaded for the language `#1'. Using the pattern for}%
\typeout{** the default language instead.}%
\else
\language=\csname l@#1\endcsname
\fi
#2}}
\providecommand{\BIBdecl}{\relax}
\BIBdecl

\bibitem{Bash2013}
B.~Bash, D.~Goeckel, and D.~Towsley, ``Limits of reliable communication with
  low probability of detection on {AWGN} channels,'' \emph{{IEEE} {J}ournal on
  {S}elected {A}reas in {C}ommunications}, vol.~31, no.~9, pp. 1921--1930,
  September 2013.

\bibitem{Bloch2015b}
M.~R. Bloch, ``Covert communication over noisy channels: A resolvability
  perspective,'' \emph{IEEE Transactions on Information Theory}, vol.~62,
  no.~5, pp. 2334--2354, May 2016.

\bibitem{Wang2016b}
L.~Wang, G.~W. Wornell, and L.~Zheng, ``Fundamental limits of communication
  with low probability of detection,'' \emph{IEEE Transactions on Information
  Theory}, vol.~62, no.~6, pp. 3493--3503, Jun. 2016.

\bibitem{Tahmasbi2017}
M.~Tahmasbi and M.~R. Bloch, ``First and second order asymptotics in covert
  communication,'' \emph{IEEE Transactions on Information Theory}, vol.~65,
  no.~4, pp. 2190 --2212, Apr. 2019.

\bibitem{Sheikholeslami2016}
A.~Sheikholeslami, B.~A. Bash, D.~Towsley, D.~Goeckel, and S.~Guha, ``Covert
  communication over classical-quantum channels,'' in \emph{Proc. of IEEE
  International Symposium on Information Theory}, Barcelona, Spain, July 2016,
  pp. 2064--2068.

\bibitem{Wang2016c}
L.~Wang, ``Optimal throughput for covert communication over a classical-quantum
  channel,'' in \emph{Proc. of IEEE Information Theory Workshop}, Sep 2016, p.
  364–368.

\bibitem{Gagatsos2020}
C.~N. Gagatsos, M.~S. Bullock, and B.~A. Bash, ``Covert capacity of bosonic
  channels,'' arXiv preprint 2002.06733, 2020.

\bibitem{Bash2017}
B.~A. Bash, C.~N. Gagatsos, A.~Datta, and S.~Guha, ``Fundamental limits of
  quantum-secure covert optical sensing,'' in \emph{Proc. of IEEE International
  Symposium on Information Theory}, Aachen, Germany, Jun. 2017, pp. 3210--3214.

\bibitem{Gagatsos2019}
C.~N. Gagatsos, B.~A. Bash, A.~Datta, Z.~Zhang, and S.~Guha, ``Covert sensing
  using floodlight illumination,'' \emph{Physical Review A}, vol.~99, p.
  062321, Jun 2019.

\bibitem{Goeckel2017}
D.~Goeckel, B.~A. Bash, A.~Sheikholeslami, S.~Guha, and D.~Towsley, ``Covert
  active sensing of linear systems,'' in \emph{Proc. of Asilomar Conference on
  Signals, Systems and Computers}, Monticello, IL, Nov. 2017, pp. 1692--1696.

\bibitem{Naghshvar2013a}
M.~Naghshvar and T.~Javidi, ``Sequentiality and adaptivity gains in active
  hypothesis testing,'' \emph{IEEE Journal of Selected Topics in Signal
  Processing}, vol.~7, no.~5, pp. 768--782, Oct. 2013.

\bibitem{Naghshvar2013}
------, ``Active sequential hypothesis testing,'' \emph{The Annals of
  Statistics}, vol.~41, no.~6, pp. 2703--2738, 12 2013.

\bibitem{Nitinawarat2013}
S.~Nitinawarat, G.~K. Atia, and V.~V. Veeravalli, ``Controlled sensing for
  multihypothesis testing,'' \emph{IEEE Transactions on Automatic and Control},
  vol.~58, no.~10, pp. 2451--2464, Oct. 2013.

\bibitem{Blahut1973}
R.~{Blahut}, ``An hypothesis-testing approach to information theory,''
  \emph{IEEE Transactions on Information Theory}, vol.~19, no.~2, p. 253, Mar.
  1973.

\bibitem{Chernoff1959}
H.~Chernoff, ``Sequential design of experiments,'' \emph{The Annals of
  Mathematical Statistics}, vol.~30, no.~3, pp. 755--770, 09 1959.

\bibitem{Hayashi2009a}
M.~Hayashi, ``Discrimination of two channels by adaptive methods and its
  application to quantum system,'' \emph{IEEE Transactions on Information
  Theory}, vol.~55, no.~8, pp. 3807--3820, Aug. 2009.

\bibitem{Naghshvar2012}
M.~{Naghshvar} and T.~{Javidi}, ``Extrinsic jensen-{Shannon} divergence with
  application in active hypothesis testing,'' in \emph{Proc. of IEEE Int Symp.
  Information Theory}, Jul. 2012, pp. 2191--2195.

\bibitem{Franceschetti2017}
M.~Franceschetti, S.~Marano, and V.~Matta, ``Chernoff test for strong-or-weak
  radar models,'' \emph{IEEE Transactions on Signal Processing}, vol.~65,
  no.~2, pp. 289--302, Jan. 2017.

\bibitem{Chiu2019}
S.~{Chiu}, N.~{Ronquillo}, and T.~{Javidi}, ``Active learning and {CSI}
  acquisition for mmwave initial alignment,'' \emph{IEEE Journal on Selected
  Areas in Communications}, vol.~37, no.~11, pp. 2474--2489, Nov. 2019.

\bibitem{Watrous_2018}
J.~Watrous, \emph{The theory of quantum information}.\hskip 1em plus 0.5em
  minus 0.4em\relax Cambridge University Press, 2018.

\bibitem{Pirandola2019}
S.~Pirandola, R.~Laurenza, C.~Lupo, and J.~L. Pereira, ``Fundamental limits to
  quantum channel discrimination,'' \emph{npj Quantum Information}, vol.~5,
  no.~1, p.~50, Dec 2019.

\bibitem{Nussbaum_Szkola_2009}
M.~Nussbaum and A.~Szkoła, ``The chernoff lower bound for symmetric quantum
  hypothesis testing,'' \emph{The Annals of Statistics}, vol.~37, no.~2, p.
  1040–1057, Apr 2009.

\bibitem{Li_2016}
K.~Li, ``Discriminating quantum states: The multiple chernoff distance,''
  \emph{The Annals of Statistics}, vol.~44, no.~4, p. 1661–1679, Aug 2016.

\bibitem{Acin_2001}
A.~Acín, ``Statistical distinguishability between unitary operations,''
  \emph{Physical Review Letters}, vol.~87, no.~17, p. 177901, Oct 2001.

\bibitem{Tahmasbi2020b}
M.~Tahmasbi and M.~R. Bloch, ``Active covert sensing,'' accepted to \emph{IEEE
  International Symposium on Information Theory}, Mar. 2020.

\bibitem{tahmasbi2018framework}
------, ``A framework for covert and secret key expansion over quantum
  channels,'' \emph{arXiv preprint arXiv:1811.05626}, 2018.

\bibitem{Fuchs_1999}
C.~Fuchs and J.~van~de Graaf, ``Cryptographic distinguishability measures for
  quantum-mechanical states,'' \emph{IEEE Transactions on Information Theory},
  vol.~45, no.~4, p. 1216–1227, May 1999.

\bibitem{shapiro2002differentiability}
A.~Shapiro, ``On differentiability of symmetric matrix valued functions,''
  \emph{School of Industrial and Systems Engineering, Georgia Institute of
  Technology}, 2002.

\bibitem{Wilde_2017}
\BIBentryALTinterwordspacing
M.~M. Wilde, \emph{Quantum information theory}, 2017. [Online]. Available:
  \url{https://doi.org/10.1017/9781316809976}
\BIBentrySTDinterwordspacing

\bibitem{Csiszar}
I.~Csiszar and J.~K{\"o}rner, \emph{Information theory: coding theorems for
  discrete memoryless systems}.\hskip 1em plus 0.5em minus 0.4em\relax
  Cambridge University Press, 2011.

\end{thebibliography}
